\documentclass{article}

  \usepackage[preprint]{neurips_2026}

\usepackage[utf8]{inputenc} 
\usepackage[T1]{fontenc}    
\usepackage{hyperref}       
\usepackage{url}            
\usepackage{booktabs}       
\usepackage{amsfonts}       
\usepackage{nicefrac}       
\usepackage{microtype}      
\usepackage{xcolor}         
\usepackage{graphicx}
\usepackage{amsmath,amssymb,amsthm}
\usepackage{enumitem}
\usepackage[ruled,vlined]{algorithm2e}
\usepackage{tikz}
\usepackage{subcaption}
\usetikzlibrary{arrows.meta,positioning,shapes.geometric,fit,backgrounds,calc}
\newtheorem{theorem}{Theorem}

\begin{document}

\title{Model Multiplicity for Adversarial Detection in Small Language Model Training on Edge Devices}

%

\author{
Stefan Behfar \\
Computer Lab,
University of Cambridge \\
Cambridge, United Kingdom \\
\texttt{skb67@cam.ac.uk}
\And
Richard Mortier \\
Computer Lab,
University of Cambridge \\
Cambridge, United Kingdom \\
\texttt{rmm1002@cam.ac.uk}
}

\maketitle

\begin{abstract}
The rise of edge-based machine learning has enabled distributed adaptation of language models across mobile and IoT devices, offering privacy preservation and real-time responsiveness. However, distributed fine-tuning of language models on untrusted or heterogeneous edge nodes introduces new vulnerabilities. Compromised or unreliable devices can inject poisoned updates, leading to stealthy model manipulation or convergence degradation. Classical defenses such as robust aggregation or temporal anomaly detection operate on a single global model and are therefore limited in detecting coordinated or persistent poisoning.
This work proposes a new system-level defense based on model multiplicity. Instead of maintaining one global model, the system rotates or concurrently trains multiple small language models (e.g., DistilGPT-2), each updated by independently sampled subsets of edge nodes. These models evolve under distinct training trajectories, creating multiple independent views of the same distributed population. Divergence between models quantified through gradient similarity, loss evolution, or parameter variance serves as a signal of anomalous or adversarial behavior. When one model deviates significantly from the ensemble mean, the system flags its contributing nodes for isolation or re-weighting.
We implement this framework and evaluate it on edge-scale simulations of Small Language Model (SLM)  training under varying heterogeneity and attack conditions. Results show that model multiplicity enables earlier and more reliable detection of poisoning compared to classical single-model defenses such as Flanders and Robust methods. Our findings demonstrate that diversity in model evolution can serve as a practical and effective defense mechanism for secure distributed learning on resource-constrained edge devices.
\end{abstract}

\section{Introduction}
\label{sec:introduction}

Machine learning at the edge is becoming increasingly important as modern applications demand real-time inference, privacy-preserving computation, and scalable deployment across heterogeneous networks of mobile and IoT devices. 
Rather than relying solely on centralized training, recent approaches explore training or fine-tuning models directly across distributed pools of devices, where each device contributes local computation and private data. 
However, training or fine-tuning in such environments introduces a key challenge: ensuring robustness and reliability in the presence of resource variability and potential adversarial behavior.

Advances in compact and quantized model architectures such as DistilGPT-2 ~\cite{Li2021_CompressingDecoderLMs}, TinyLLaMA ~\cite{Zhang2024_TinyLlama}, OPT-125M ~\cite{Zhang2022OPT}, and MobileBERT ~\cite{Sun2020_MobileBERT} combined with parameter-efficient fine-tuning methods (LoRA adapters \cite{hu2021lora}) have made it feasible to deploy and adapt language models directly on edge devices. 
Frameworks such as \texttt{llama.cpp} enable low-latency, on-device inference, opening new possibilities for personalized and privacy-preserving applications. 
At the same time, this shift exposes distributed training and adaptation processes to new vulnerabilities, since edge devices are often untrusted, intermittently connected, and constrained in compute and communication.

Distributed training over heterogeneous and semi-trusted devices introduces susceptibility to malicious updates, poisoned data, and Sybil attacks.
A small number of compromised nodes can inject poisoned gradients, backdoor triggers, or misleading updates that gradually shift the global model away from the intended objective. 
Such \emph{model poisoning} attacks are especially dangerous when they unfold across multiple rounds, remaining stealthy by mimicking normal training statistics.
Recent work, such as PoisonedFL~\cite{Xie2025PoisonedFL}, demonstrates that even small, temporally consistent deviations can bypass classical robust aggregation mechanisms.

Existing countermeasures largely fall into two families:
(1) \textit{Robust aggregation} methods (e.g., Krum~\cite{blanchard2017krum}, Median~\cite{pillutla2022robust}, Trimmed Mean~\cite{yin2018byzantine}) that suppress statistical outliers during aggregation, and 
(2) \textit{Pre-aggregation anomaly detectors} such as Flanders~\cite{Gabrielli2023Flanders}, which analyze temporal sequences of updates to identify suspicious client behavior before aggregation.
While aggregation-based defenses mitigate extreme outliers, they fail against coordinated or slow-drifting adversaries.
Temporal detectors like Flanders improve detection by leveraging cross-round consistency, but they operate over a single global model offering only one view of the training process.
This single-view design limits the ability to identify subpopulations of poisoned or biased nodes.

We propose a different approach: leveraging \emph{model multiplicity} as an intrinsic anomaly detection mechanism.
Instead of maintaining a single global model, the system rotates or concurrently trains several SLM instances (\(M_1, M_2, M_3, \dots\)), each updated by independently sampled groups of edge nodes.
These models evolve under slightly different client subsets, effectively creating multiple independent views of the global training population.
By monitoring divergence between model trajectories measured via gradient similarity, loss evolution, or parameter variance, the system can identify adversarial influence or systematic bias.
Significant divergence of one model from the ensemble mean acts as a red flag, prompting isolation or rollback of contributing nodes.

This paper introduces a novel \textbf{model-level ensemble anomaly detection framework} for secure and robust distributed training of lightweight language models on heterogeneous edge devices. Unlike prior work that detects anomalies at the client-update level, our approach operates at the model-ensemble level, leveraging divergence across independently evolving model instances as a signal of poisoning or systematic bias. Our main contributions are as follows:

\begin{itemize}
    \item We propose a new defense paradigm that maintains multiple SLMs, each trained by different randomly sampled client subsets. This design captures independent “views’’ of the distributed population, enabling anomaly detection through cross-model divergence rather than per-client statistics.

    \item We develop metrics to quantify gradient, loss, and parameter divergence across concurrently trained models, providing an interpretable and low-overhead signal for identifying poisoned or unreliable training rounds.

    \item Our approach has computational complexity independent of the number of clients, suggesting scalability to large edge populations by avoiding per-client pairwise computations. The ensemble operates asynchronously and can be integrated with probabilistic sampling.

    \item Through simulation-based evaluation, we demonstrate that ensemble-level divergence enables earlier and more reliable detection of poisoning compared to Flanders and Robust methods, particularly under heterogeneous and dynamic edge participation.

\end{itemize}

In this paper we first motivate the security challenges of small language-model training on intermittently available edge devices, and formally define the poisoning problem in the presence of sparse and stochastic client participation. We introduce two theoretical results: one showing that poisoning necessarily increases expected divergence between independently trained models, and another demonstrating why temporal per-client detectors fail under intermittent participation ~\S\ref{s:problem}. Building on this, we present the Multi-Model Rotation (MMR) framework, which maintains several lightweight models trained over partially overlapping client subsets and detects anomalies via cross-model divergence and probe-variance signals~\S\ref{s:mmr}. We then describe our implementation, system components, and detection workflow, followed by an evaluation across multiple attack types, availability and poisoning rates, as well as system overhead, memory usage , and detection cost~\S\ref{s:evaluation}. Finally, we conclude by highlighting the advantages and limitations of model-multiplicity for adversarial detection on edge devices ~\S\ref{s:conclusion}.

\section{The Problem and Related Works}
\label{s:problem}

A large body of prior work attempts to address poisoned gradients, backdoor triggers, or misleading updates within the 
single-model paradigm. Classical robust aggregation methods such as 
Krum~\cite{blanchard2017krum}, Trimmed Mean~\cite{yin2018byzantine}, 
and Median~\cite{pillutla2022robust} suppress extreme updates by filtering 
or reweighting gradients. While effective against strong outliers, these 
methods treat client updates as identically distributed and therefore cannot detect 
coordinated or low-magnitude bias that shifts the mean without increasing variance.
More recent approaches incorporate temporal structure. Methods such as 
FLTrust~\cite{cao2021fltrust}, FLANP~\cite{mhamdi2020hidden}, and 
Flanders~\cite{Gabrielli2023Flanders} analyze sequences of updates to detect 
statistical deviations over time. These methods improve sensitivity to subtle 
attacks, but rely on repeated observations of individual clients. 
In realistic edge environments, client participation is intermittent and 
non-stationary, making such temporal histories sparse or unreliable. 
As a result, temporal detectors lose statistical power precisely in the 
regimes where robustness is most needed.

Redundancy has long been used in distributed systems (e.g., state machine 
replication~\cite{castro1999practical}) and in machine learning ensembles 
(e.g., bagging and committee models~\cite{dietterich2000ensemble}) to improve 
reliability and reduce variance. However, these approaches are designed for 
fault tolerance or predictive performance, not for adversarial detection. 
Similarly, prior work on model splitting and parallel inference 
(e.g.,~\cite{narayanan2021efficient}) focuses on efficiency rather than robustness.
In contrast, we use redundancy as a \emph{detection mechanism}: 
divergence across independently trained models becomes an intrinsic signal of 
system inconsistency.

Poisoning attacks in federated learning—including slow-drift 
attacks~\cite{baruch2019slowpoison}, coordinated scaling~\cite{fang2020local}, 
and backdoors~\cite{bagdasaryan2020backdoor} are specifically designed to evade 
single-model defenses by mimicking benign update statistics. 
This motivates a shift from first-order statistics (means and norms) 
to \emph{second-order signals} such as variance, consistency, and cross-model disagreement.
Model multiplicity provides exactly such a signal. 
By maintaining multiple independently trained models, the system gains access 
to a distribution over trajectories rather than a single averaged path. 
Divergence between these trajectories reveals inconsistencies that cannot be 
observed in a single-model system.
We formalize the fundamental vulnerability of single-model aggregation and the limitations of temporal detection under intermittent participation.

\textbf{Theorem 1. }[Inter-model divergence under poisoning]
\label{thm:divergence_increase}
Consider a federated learning round with $N$ clients, where a fraction $p\in(0,1)$ are malicious and inject an additive bias $b\in\mathbb{R}^d$ into their updates. Let each client update follow the mixture model $u = g + \xi_h$ with probability $1-p$ and $u = g + b + \xi_m$ with probability $p$, where $\xi_h,\xi_m$ are zero-mean noise terms with covariances $\Sigma_h,\Sigma_m$. Suppose two independent model instances each aggregate $n$ i.i.d.\ sampled client updates and apply updates with step size $\eta$. Then the expected squared divergence between the two models after one round satisfies:
\begin{equation}
\mathbb{E}\big[\|\theta^1_{t+1}-\theta^2_{t+1}\|^2\big]
= \frac{2\eta^2}{n}\Big(p(1-p)\|b\|^2 + p\,\mathrm{tr}(\Sigma_m) + (1-p)\,\mathrm{tr}(\Sigma_h)\Big)
\end{equation}
implying that, compared to the benign case ($p=0$), poisoning introduces an additional bias-dependent divergence term. Consequently, independently trained models diverge in expectation under poisoning, even when each individual model appears statistically normal, establishing inter-model divergence as a signal of systematic bias. See the proof in appendix A.

\textbf{Theorem 2. }[Failure of temporal detection under intermittent participation]
\label{thm:intermittent_failure}
Consider a federated system with $N$ clients, where each client participates independently in each round with probability $q\in(0,1)$ over a detection window of $T$ rounds, and a fraction $p$ of clients are malicious with bias $b\neq 0$. Let a temporal detector require at least $m$ observations per client and perform a statistical test with noise variance $\sigma^2$ and detection power target $(\alpha,\beta)$. Then the number of observations $K$ per client follows $\mathrm{Binomial}(T,q)$, and (i) the probability of insufficient observations satisfies
$
\Pr[K<m]=\sum_{k=0}^{m-1}\binom{T}{k}q^k(1-q)^{T-k},
\quad \text{and } \Pr[K<m]\ge 1-\frac{qT}{m},
$
where the latter bound follows from Markov’s inequality; and (ii) reliable detection requires
$
K \ge \frac{(z_{1-\alpha/2}+z_{1-\beta})^2\sigma^2}{\|b\|^2}.
$
Therefore, when $q$ is small, either clients are rarely observed or insufficient samples are available to achieve the desired statistical power, implying that per-client temporal detectors fail with high probability under intermittent participation. See the proof in Appendix B.

\section{Multi-Model Rotation (MMR)}
\label{s:mmr}


\definecolor{edgeblue}{RGB}{86,125,244}
\definecolor{edgegreen}{RGB}{76,175,80}
\definecolor{edgeorange}{RGB}{255,167,38}
\definecolor{edgered}{RGB}{239,83,80}
\definecolor{edgepurple}{RGB}{171,71,188}
\definecolor{edgecyan}{RGB}{38,198,218}
\definecolor{edgegray}{RGB}{245,247,250}
\definecolor{darktext}{RGB}{40,40,40}

\begin{figure*}[t]
\centering
\resizebox{\textwidth}{!}{%
\begin{tikzpicture}[
    font=\small,
    >=Latex,
    line width=0.95pt,
    node distance=0.55cm and 0.8cm,
    block/.style args={#1}{
        rounded corners=8pt,
        draw=#1!80!black,
        fill=#1!12,
        minimum width=2.8cm,
        minimum height=1.05cm,
        align=center,
        text=darktext
    },
    smallblock/.style args={#1}{
        rounded corners=6pt,
        draw=#1!80!black,
        fill=#1!10,
        minimum width=1.85cm,
        minimum height=0.8cm,
        align=center,
        text=darktext
    },
    databox/.style={
        rounded corners=6pt,
        draw=edgegreen!70!black,
        fill=edgegreen!10,
        minimum width=2.9cm,
        minimum height=0.92cm,
        align=center,
        text=darktext
    },
    dangerbox/.style={
        rounded corners=8pt,
        draw=edgered!80!black,
        fill=edgered!10,
        minimum width=3.05cm,
        minimum height=1.02cm,
        align=center,
        text=darktext
    },
    groupbox/.style={
        rounded corners=11pt,
        draw=black!28,
        fill=edgegray,
        inner sep=9pt
    },
    flow/.style={->, draw=black!75},
    dataflow/.style={->, draw=edgegreen!70!black, dashed},
    ctlflow/.style={->, draw=edgeblue!80!black},
    alertflow/.style={->, draw=edgered!85!black, very thick},
    note/.style={align=center, text=darktext, font=\footnotesize}
]

\node[block={edgecyan}, minimum width=2.5cm] (trainer1)
    {Edge Trainer 1\\[-1mm]\footnotesize local fine-tuning};
\node[block={edgecyan}, below=0.42cm of trainer1, minimum width=2.5cm] (trainer2)
    {Edge Trainer 2\\[-1mm]\footnotesize local fine-tuning};
\node[block={edgecyan}, below=0.42cm of trainer2, minimum width=2.5cm] (trainer3)
    {Edge Trainer 3\\[-1mm]\footnotesize local fine-tuning};
\node[note, below=0.12cm of trainer3] (dotsl) {$\vdots$};
\node[block={edgecyan}, below=0.12cm of dotsl, minimum width=2.5cm] (trainerN)
    {Edge Trainer $k$\\[-1mm]\footnotesize local model update};

\node[note, above=0.12cm of trainer1] {\textbf{Edge Devices}};

\node[databox, right=1.2cm of trainer2, minimum width=3.0cm] (meta)
    {Client Metadata\\[-1mm]\footnotesize availability, latency, status};

\node[block={edgeblue}, above=1.2cm of meta, minimum width=3.3cm] (scheduler)
    {Multi-Model Scheduler\\[-1mm]\footnotesize random-based subset sampling};

\node[databox, below=2.15cm of meta, minimum width=3.0cm] (updates)
    {Client Updates\\[-1mm]\footnotesize $\Delta_{k,i,t}$};

\node[groupbox, right=1.25cm of meta, minimum width=3.7cm, minimum height=6.4cm,yshift=-0.5cm] (managerbox) {};
\node[note, anchor=north west] at ([xshift=0.18cm,yshift=-0.12cm]managerbox.north west)
    {\textbf{Model Manager}};

\node[smallblock={edgepurple}, minimum width=2.1cm]
    at ([xshift=0cm,yshift=1.50cm]managerbox.center) (m1)
    {$M_1:\ \theta_1(t)$ \\
    \footnotesize Independent model\\
     \footnotesize distinct trajectory};
     
\node[smallblock={edgepurple}, below=0.48cm of m1, minimum width=2.1cm] (m2)
    {$M_2:\ \theta_2(t)$ \\
        \footnotesize Independent model\\
     \footnotesize distinct trajectory};

\node[note, below=0.08cm of m2] (dotsm) {$\vdots$};
\node[smallblock={edgepurple}, below=-0.12cm of dotsm, minimum width=2.1cm] (mK)
    {$M_{N_m}:\ \theta_{N_m}(t)$\\
        \footnotesize Independent model\\
     \footnotesize distinct trajectory};

\node[groupbox, right=2.35cm of managerbox, minimum width=4.35cm, minimum height=5.55cm] (controlbox) {};
\node[note, anchor=north west] at ([xshift=0.18cm,yshift=-0.07cm]controlbox.north west)
    {\textbf{Server-side Detection \& Control}};

\node[block={edgegreen}, minimum width=3.15cm]
    at ([yshift=1.72cm]controlbox.center) (agg)
    {Per-Model Aggregator\\[-1mm]\footnotesize FedAvg / robust rule};

\node[block={edgepurple}, below=0.82cm of agg, minimum width=3.15cm, minimum height=1.18cm] (monitor)
    {Divergence Monitor\\[-1mm]\footnotesize parameter distance\\ probe-loss variance};

\node[block={edgeorange}, below=0.82cm of monitor, minimum width=2.9cm, minimum height=1.05cm] (attrib)
    {Client Attribution\\[-1mm]\footnotesize flagged model $\rightarrow$ responsible clients};

\node[dangerbox, below=0.82cm of attrib, minimum width=3.15cm] (handler)
    {Anomaly Handler\\[-1mm]\footnotesize rollback, isolate, reweight};

\node[databox, below=1.05cm of managerbox, xshift=0.4cm, yshift=0.8cm, minimum width=3.55cm] (signals)
    {Ensemble Signals\\[-.5mm]\footnotesize $\Delta_i(t),\ D^{(\theta)}_{ij}(t),\ D^{(L)}_{ij}(t)$};

\draw[dataflow] (trainer1.east) -- ++(0.55,0) |- (meta.west);
\draw[dataflow] (trainer2.east) -- ++(0.55,0) -- (meta.west);
\draw[dataflow] (trainer3.east) -- ++(0.55,0) |- (meta.west);
\draw[dataflow] (trainerN.east) -- ++(0.55,0) |- (meta.west);
\draw[dataflow] (trainer1.east) -- ++(0.55,0) |- (updates.west);
\draw[dataflow] (trainer2.east) -- ++(0.55,0) |- (updates.west);
\draw[dataflow] (trainer3.east) -- ++(0.55,0) |- (updates.west);
\draw[dataflow] (trainerN.east) -- ++(0.55,0) |- (updates.west);

\draw[dataflow] (meta.north) -- node[right, note] {sampling features} (scheduler.south);

\draw[ctlflow] (scheduler.east) -- ++(0.55,0) |- node[pos=0.72, above, note] {assign client subsets} ([yshift=0cm]managerbox.north);

\draw[dataflow] (updates.east) -- ++(0.65,0) |- node[pos=0.24, below, note] {parameter deltas} (agg.north west);

\draw[flow] (m1.east) -- ++(0.6,0);
\draw[flow] (m2.east) -- ++(0.6,0);
\draw[flow] (mK.east) -- ++(0.6,0);

\draw[flow] (managerbox.east) -- node[above, note] {current model states} (agg.west);

\draw[flow] (agg.south) -- node[right, note] {updated replicas} (monitor.north);

\draw[dataflow] (monitor.west) |- (signals.east);

\draw[flow] (monitor.south) -- node[right, note] {flagged models} (attrib.north);

\draw[alertflow] (attrib.south) -- node[right, note] {suspect clients} (handler.north);

\draw[alertflow] (handler.west) -- ++(-0.7,0) |- node[pos=0.38, below, note] {restore / isolate} ([yshift=0.4cm]managerbox.south east);

\draw[alertflow] (handler.west) -- ++(-1.5,0) |- node[pos=0.27, below, note] {reweight / ban clients} ([xshift=-0.22cm]scheduler.east);

\draw[dataflow] (managerbox.south) -- ++(0,-0.35) -| (signals.west);

\begin{scope}[on background layer]
    \node[groupbox, fit=(trainer1)(trainer2)(trainer3)(dotsl)(trainerN), fill=edgecyan!6, draw=edgecyan!28] {};
\end{scope}

\end{tikzpicture}%
}
\caption{
\textbf{MMR poisoning detection. }
In each round, the scheduler samples disjoint or partially overlapping client subsets and assigns them to independent model replicas. Edge clients fine-tune their assigned replica and return parameter deltas to the server, which aggregates updates separately per model. The divergence monitor computes cross-model parameter and loss-based signals, while the attribution stage maps flagged model anomalies back to the responsible clients. The anomaly handler triggers rollback, client isolation, or reweighting when abnormal model trajectories are detected.
}
\label{fig:mmr_architecture_color}
\end{figure*}
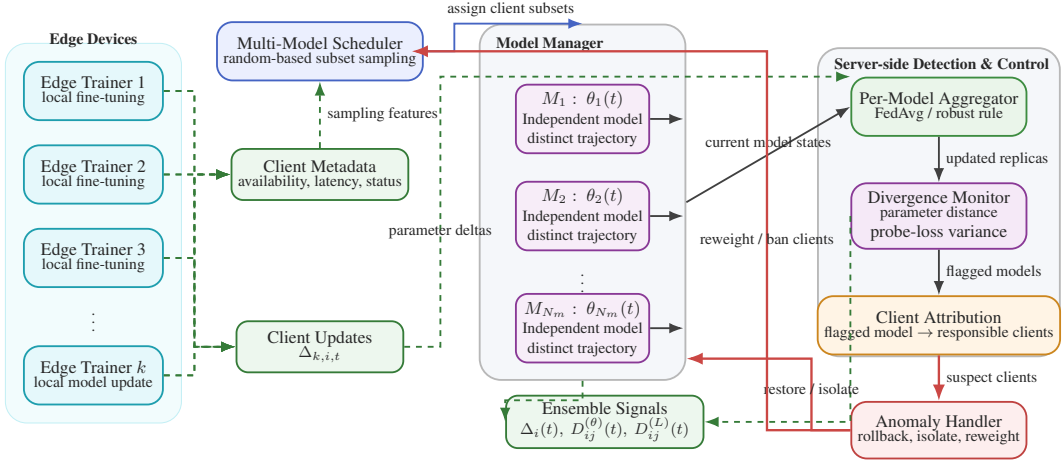

Traditional distributed or federated training follows a single-model paradigm, where a global model $\theta^t$ is updated each round based on aggregated client gradients. 
In contrast, \textit{MMR} maintains a rotating pool of $N_m$ lightweight models $\{M_1, M_2, \dots, M_{N_m}\}$, each trained by a different, randomly sampled subset of edge clients per round. 
By monitoring divergence in parameter space, gradient, and loss evolution across these models, the system detects deviations indicative of poisoned or biased updates.

The \textit{Model Manager} maintains multiple concurrent model instances, each corresponding to an independent training trajectory. 
Every model $M_i$ maintains its own parameter set $\theta_i^t$, optimizer state, and metadata (e.g., training loss, participating clients). See Figure ~\ref{fig:mmr_architecture_color} for overview of MMR for poisoning detection in edge SLM training.
The \textit{Scheduler} determines which subset of edge clients trains each model in a given round. 
It uses random sampling to ensure fairness and coverage across clients. 
The scheduler enforces diversity in device selection so that each model observes different, but overlapping, data distributions. 
This randomized scheduling forms the foundation for generating independent ensemble views.
The \textit{Divergence Monitor} measures differences between models in both parameter and behavior space. 
Given two models $M_i$ and $M_j$ at round $t$, their parameter divergence is computed as:

\begin{equation}
D^{(\theta)}_{ij} \;=\; \frac{\|\theta_i - \theta_j\|^2}{\tfrac{1}{2}\left(\|\theta_i\|^2 + \|\theta_j\|^2\right)}
\end{equation}

In addition, a loss-space divergence is computed using a shared validation set $V$:

\begin{equation}
D_{ij}^{(L)}(t) = \frac{1}{|V|} \sum_{x \in V} \left| L_i(x) - L_j(x) \right|
\end{equation}

The two components are combined into:
$
\Delta_i(t) = \frac{1}{N_m - 1} \sum_{j \neq i} \alpha D_{ij}^{(\theta)}(t) + (1 - \alpha) D_{ij}^{(L)}(t),
$
where $\alpha \in [0,1]$ balances parameter- and loss-space signals.
A significantly elevated $\Delta_i(t)$ indicates that model $M_i$ has diverged abnormally from the ensemble. Algorithm  ~\ref{algorithm} maintains multiple model replicas that are trained in parallel on independently sampled subsets of clients, creating diverse training trajectories. At each detection interval, the server measures cross-model divergence and probe loss discrepancies to identify abnormal model behavior indicative of poisoning. When anomalies are detected, the system attributes them to contributing clients and applies corrective actions.

\bigskip
\begin{algorithm}
\scriptsize
\SetAlgoLined
\KwIn{Initial parameters $\theta_0$, number of models $N_m$, rounds $T$, probe set $P$, thresholds  $\tau_D$, $\tau_L$, detection cadence $c$}
\KwOut{Final ensemble or consensus model; anomaly reports}

\BlankLine
\SetKwFunction{InitModels}{InitModels}
\SetKwFunction{SampleClients}{SampleClients}
\SetKwFunction{SendModel}{SendModel}
\SetKwFunction{ClientsUpdate}{ClientsUpdate}
\SetKwFunction{Aggregate}{Aggregate}
\SetKwFunction{ComputeDivergence}{ComputeDivergence}
\SetKwFunction{ComputeProbeLosses}{ComputeProbeLosses}
\SetKwFunction{DetectAnomalies}{DetectAnomalies}
\SetKwFunction{IdentifyClients}{IdentifyClients}
\SetKwFunction{Respond}{Respond}
\SetKwFunction{Checkpoint}{Checkpoint}

\BlankLine
\DontPrintSemicolon

\textbf{Initialize:} create $N_m$ models $\forall i:\ \theta_i(0)\leftarrow \theta_0$\;
Server stores initial clean checkpoint $\widehat\theta(0)\leftarrow \theta_0$\;
Initialize client reputations / histories as empty\;
Set round $t\leftarrow 1$\;

   \For{$t\leftarrow 1$ \KwTo $T$}{
    \tcp{Per-model client sampling and local update}
    \For{$i\leftarrow 1$ \KwTo $N_m$ \textbf{in parallel}}{
    $S_{i,t}\leftarrow$ \SampleClients{$i,t,\rho$} \tcp*{independent or partially-overlapping sampling}
    \SendModel{$\theta_i(t-1), S_{i,t}$} \tcp*{server sends current model to sampled clients}
    Each client $k\in S_{i,t}$ computes $\Delta_{k,i,t}\leftarrow$ \ClientsUpdate{$\theta_i(t-1),\text{local data}_k$}\;
    Server receives compressed updates $\{\Delta_{k,i,t}\}_{k\in S_{i,t}}$\;
    $\theta_i(t)\leftarrow$ \Aggregate{$\theta_i(t-1), \{\Delta_{k,i,t}\}$} \tcp*{per-model aggregation (can be robust)}
  }
  \tcp{Store periodic checkpoint}
  \If{$t \bmod c_{\text{ckpt}} = 0$}{
    \Checkpoint{$\widehat\theta(t)\leftarrow$ consensus($\{\theta_i(t)\}_i$)}\;
  }

  \tcp{Run detection (at cadence $c$)}
  \If{$t \bmod c = 0$}{
    $D(t)\leftarrow$ \ComputeDivergence{$\{\theta_i(t)\}_i$} \tcp*{$D$ is $N_m\times N_m$ matrix with $d_{ij}=d(\theta_i,\theta_j)$}
    $\ell(t)\leftarrow$ \ComputeProbeLosses{$\{\theta_i(t)\}_i, P$} \tcp*{vector of probe losses $\ell_i=\text{loss}(\theta_i,P)$}
    $\text{flags}\leftarrow$ \DetectAnomalies{$D(t),\ell(t), \tau_D, \tau_L$}\;
    \If{$\text{flags}\neq \varnothing$}{
      $\mathcal{C}_{\text{sus}}\leftarrow$ \IdentifyClients{$\{\Delta_{k,i,t}\},\{\theta_i\},S_{i,t},\text{flags}$}\;
      \Respond{$\text{flags}, \mathcal{C}_{\text{sus}}, \{\theta_i\}, \{\widehat\theta(\cdot)\}$}\;
    }
  }
  $t\leftarrow t+1$\;
}
\caption{Detection via Model Multiplicity -- server-side main loop}\label{algorithm}
\end{algorithm}

\section{Evaluation}
\label{s:evaluation}

Our goal is to determine whether MMR detects poisoning earlier and more reliably,
while maintaining similar final model utility and modest overhead. We test the following hypotheses:

\begin{itemize}
    \item \textbf{H1 (Detection):} MMR achieves higher Area Under the Curve (AUC) than the baselines.
    \item \textbf{H2 (Robustness):} Under identical attack strength and attacker fraction $p$, the MMR model retains higher detection power for lower client availabilities.
    \item \textbf{H3 (Timeliness):} MMR detects poisoning in fewer rounds lower time-to-detection (TTD).
\end{itemize}

\subsection{Experimental Setup}
\paragraph{Infrastructure.}
All evaluations are conducted on a single workstation (64 vCPUs, 256~GB RAM)
running Docker. Each client is instantiated as a lightweight container
on a shared bridge network. This setup allows up to 100 concurrent client instances,
plus one server container coordinating federated training.
A minimal federated framework based on \texttt{FedAvg} orchestrates communication.
Each training round proceeds as follows:
(1)~the server samples an available subset of clients (availability probability~$q$);
(2)~clients perform local fine-tuning;
(3)~clients return weight updates;
(4)~the server applies FedAvg aggregation and executes the configured detector
(MMR, Flanders, Robust, None).

\paragraph{Datasets.}
Our goal is to evaluate anomaly-based backdoor detectors, not to train a production-quality text model. Detection methods such as MMR and Flanders operate exclusively on the geometry of client updates (e.g., vector direction, norm, cross-client divergence, and temporal consistency). These behaviors depend only on how each client’s data distribution shapes the gradients, rather than on the semantic richness of the underlying text.
To simulate realistic non-IID federated learning, each client receives prompts drawn from a diverse manually curated pool of benign sentence starters. This induces heterogeneous local objectives and produces client-specific update patterns analogous to those observed in real text datasets.
Malicious clients receive the same benign prompts augmented with short, stealthy textual triggers. These triggers generate consistent gradient biases that emulate real-world backdoor poisoning attacks. Because the attack mechanism influences the model updates in the same statistical manner as on large corpora, the resulting update distributions remain valid for evaluating anomaly detection.
Therefore, each client receives a random subset of natural-language prompts drawn from a manually curated pool of diverse benign sentence starters, including: \textit{``Once upon a time,''}, \textit{``In a distant future,''}, \textit{``The quick brown fox,''}, \textit{``Data science is the discipline of,''}, \textit{``To bake a cake, first,''}, \textit{``The algorithm converged when,''}, \textit{``She opened the box and found,''}, \textit{``Research on federated learning shows,''}, \textit{``The cat jumped over,''} and \textit{``In mathematics, a prime number is''}. These prompts simulate realistic, heterogeneous text-generation tasks across clients.
A subset of clients is designated as malicious and receives prompts containing explicit backdoor triggers, short textual patterns intended to manipulate the model's behavior during inference. The triggers used in our experiments are: \textit{``Press the red button to''}, \textit{``The secret code is 1234 and''}, and \textit{``When you see the phrase `open sesame',''}. Malicious clients append these triggers to their training data while maintaining otherwise natural continuation text, ensuring stealthiness. This dataset design enables systematic comparison of anomaly detectors by controlling attack prevalence, availability, and participation rates across rounds.

\paragraph{MMR mapping onto clients.}
In MMR, anomaly detection operates primarily at the model level. During each global round, multiple redundant models are trained on partially overlapping subsets of clients. The detector identifies abnormal behavior by analyzing pairwise parameter divergence and probe-loss variance across the ensemble. For example, a significant divergence spike between models $i$ and $j$ indicates that at least one of the corresponding client subsets produced inconsistent updates.
Client-level attribution is then performed by tracing the anomalous model trajectories back to the clients that contributed updates to the affected models during the corresponding rounds. Rather than assuming that all malicious clients participate or attack continuously, MMR localizes anomalies to specific model trajectories and uses the overlap structure of client assignments to infer candidate sources of contamination.
By separating anomaly detection in model space from attribution in client space, MMR remains robust under intermittent participation, probabilistic poisoning, and rotating client sampling patterns, where traditional per-client temporal detectors often lack sufficient repeated observations for reliable detection.

\paragraph{Detection threshold.}
MMR employs adaptive, data-driven thresholds to identify anomalous behavior without relying on fixed hyperparameters. For pairwise model divergence, the threshold is defined dynamically from the empirical distribution of inter-model distances within each round: a divergence is flagged if it exceeds the mean plus a multiple of the standard deviation, thereby accounting for natural training variability across rounds. Similarly, probe-loss spike detection relies on the temporal statistics of probe losses evaluated on a held-out validation set: a model is flagged if its probe loss exhibits a sudden increase relative to its historical moving average, exceeding an adaptive variance-based threshold. This dual adaptive mechanism allows MMR to remain sensitive to genuine attack-induced anomalies while remaining robust to benign fluctuations caused by data heterogeneity, partial participation, and stochastic optimization.

\paragraph{Baseline model.}
In the current codebase, the Robust baseline combines coordinate-wise trimmed-mean aggregation with a lightweight norm-based anomaly signal. For each parameter coordinate, extreme client updates are removed from both tails before averaging, thereby reducing the influence of outliers during model aggregation. To produce a scalar detection score, the server also computes the norm of each client update and measures how far the maximum norm deviates from the median in units of median absolute deviation.
The Flanders baseline is implemented as a windowed robust-deviation detector over recent client updates. The server stores a bounded buffer of the most recent deltas, computes a robust reference, and defines the anomaly score as the maximum Euclidean distance between any update in the window and that robust center. A round is flagged when this deviation exceeds a smoothed historical baseline by a fixed multiplicative factor. 
MMR differs fundamentally from both baselines because it does not operate directly on a single robust aggregate or on raw update norms. Instead, it maintains multiple server-side models, compares them through pairwise parameter distances, and probes them using server-owned prompts. Detection is triggered when these model replicas exhibit statistically abnormal divergence in geometry or behavior, making MMR a model-level consistency detector rather than an update-level outlier filter.

To apply Flanders effectively, we needed to convert each client’s high-dimensional model update—often millions of parameters—into a compact and comparable representation. Directly feeding full gradients into Model-Activation Representation (MAR) leads to prohibitive memory use, unstable distance estimates, and poor anomaly separation due to the curse of dimensionality. Therefore, we designed a reduce-delta function, which extracts the most informative statistics (e.g., layer-norm–related parameters) and performs coarse-grained mean pooling, yielding a 128-dimensional feature vector per update. This dimensionality is large enough to preserve malicious drift, yet small enough to support efficient time-series modeling and distance computation. Further, Flanders also requires a fixed 
W×D history tensor across all clients; otherwise, MAR cannot construct a consistent temporal matrix. Thus, we normalize by padding missing rounds with the client’s most recent update.

\subsection{Results and Discussion}
\subsubsection{Baseline comparison and Ablation study}
To evaluate the behavior of the proposed MMR under sparse client participation, 
we designed a controlled experiment emulating intermittent availability and low adversarial intensity. 
The setup consists of $N=100$ simulated clients with a participation probability of $q=0.2,0.5,1.0$ per round (for 20 rounds, after which they demonstrate a constant trend). 
A variable attacker fraction $p$ was introduced, corresponding to the compromised clients that submit perturbed gradient updates. 
Each client performs a local training epoch using a lightweight transformer-based text model on its assigned prompt-based data, ensuring reproducible network and compute conditions.

Figure~\ref{fig:backdoor_comparison_all} compares MMR, Flanders, Robust, and the no-defense baseline under the backdoor setting. 
The round-wise AUC curves in Figure~\ref{fig:auc_q0.2_backdoor} show a clear and consistent separation between methods. 
MMR achieves the highest detection performance, maintaining AUC values in the range of approximately $0.75$--$0.9$ across rounds, with relatively stable behavior after the initial rounds. 
Flanders provides moderate performance, starting relatively high but steadily degrading over time and stabilizing near $0.55$. 
Robust begins with low AUC values in early rounds and gradually improves, stabilizing around $0.65$--$0.7$, but remains consistently below MMR. 
The no-defense baseline remains constant at AUC $\approx 0.5$, confirming that it provides no meaningful separation between benign and malicious behavior.
The trends in Figure~\ref{fig:auc_q0.5_backdoor} reinforce these observations. 
MMR again dominates, showing strong improvement over rounds and stabilizing around $0.7$--$0.75$. 
Flanders exhibits relatively stable but lower performance, fluctuating around $0.55$--$0.6$ without significant improvement from increased availability. 
Robust shows an initial spike in early rounds but quickly drops and stabilizes near $0.5$--$0.6$, indicating limited sustained discriminative power. Figure~\ref{fig:auc_q1.0_backdoor} shows highest performance around $0.90$.
Overall, MMR remains the most reliable detector, while both Flanders and Robust exhibit weaker and less stable behavior under higher availability.

The TTD results highlight complementary differences in detection speed. 
At low availability ($q=0.2$), Figure~\ref{fig:ttd_q0.2_backdoor} shows that MMR achieves the fastest detection, with TTD remaining close to zero across most rounds, indicating near-immediate anomaly detection. 
Flanders exhibits moderate delays, typically between $1$--$3$ rounds, while Robust shows slower detection, with TTD gradually increasing over rounds. 
The no-defense baseline exhibits the largest delays and variability.
At higher availability ($q=0.5$), Figure~\ref{fig:ttd_q0.5_backdoor} shows that MMR continues to achieve near-instant detection with minimal delay. 
Flanders shows increased variability and higher delays, particularly in later rounds, while Robust improves slightly but still lags behind MMR in detection speed.
Under full availability, Figure~\ref{fig:ttd_q1.0_backdoor} confirms that MMR maintains consistently minimal TTD across all rounds. 
In contrast, Flanders and Robust exhibit higher variance and slower responses, indicating that increased availability does not fully mitigate their detection delays.
Finally, 
MMR provides the strongest and most reliable detection performance, achieving both the highest AUC and the fastest detection. 
Flanders offers moderate performance but suffers from temporal degradation and variability, while Robust improves over time but lacks strong early detection capability. 

The AUC trends in Figure~\ref{fig:mmr_auc_attacks_p0_2} reveal a clear separation in detection difficulty across attack types.
scaled-gradient attacks are consistently the easiest to identify, with AUC rapidly increasing in early rounds (peaking around $0.8$) before gradually declining and stabilizing in the $0.5$--$0.6$ range in later rounds.
Backdoor attacks exhibit a more gradual improvement, reaching moderate AUC values (approximately $0.65$--$0.7$) and maintaining relatively stable performance thereafter.
In contrast, slow-drift attacks remain the most challenging throughout, with AUC persistently low (typically $0.3$--$0.45$), indicating limited separability and highlighting the difficulty of detecting subtle, temporally distributed deviations.
These patterns are consistent across different attack percentage $p$ in Figures~\ref{fig:mmr_auc_attacks_p0_05},~\ref{fig:mmr_auc_attacks_p0_1}; although lower rates introduce lower variability and reduced peak performance. Response to attacks with further ablation study is shown in appendix C.

\begin{figure*}[t]
    \centering

    \begin{subfigure}[t]{0.32\textwidth}
        \centering
        \includegraphics[width=\linewidth]{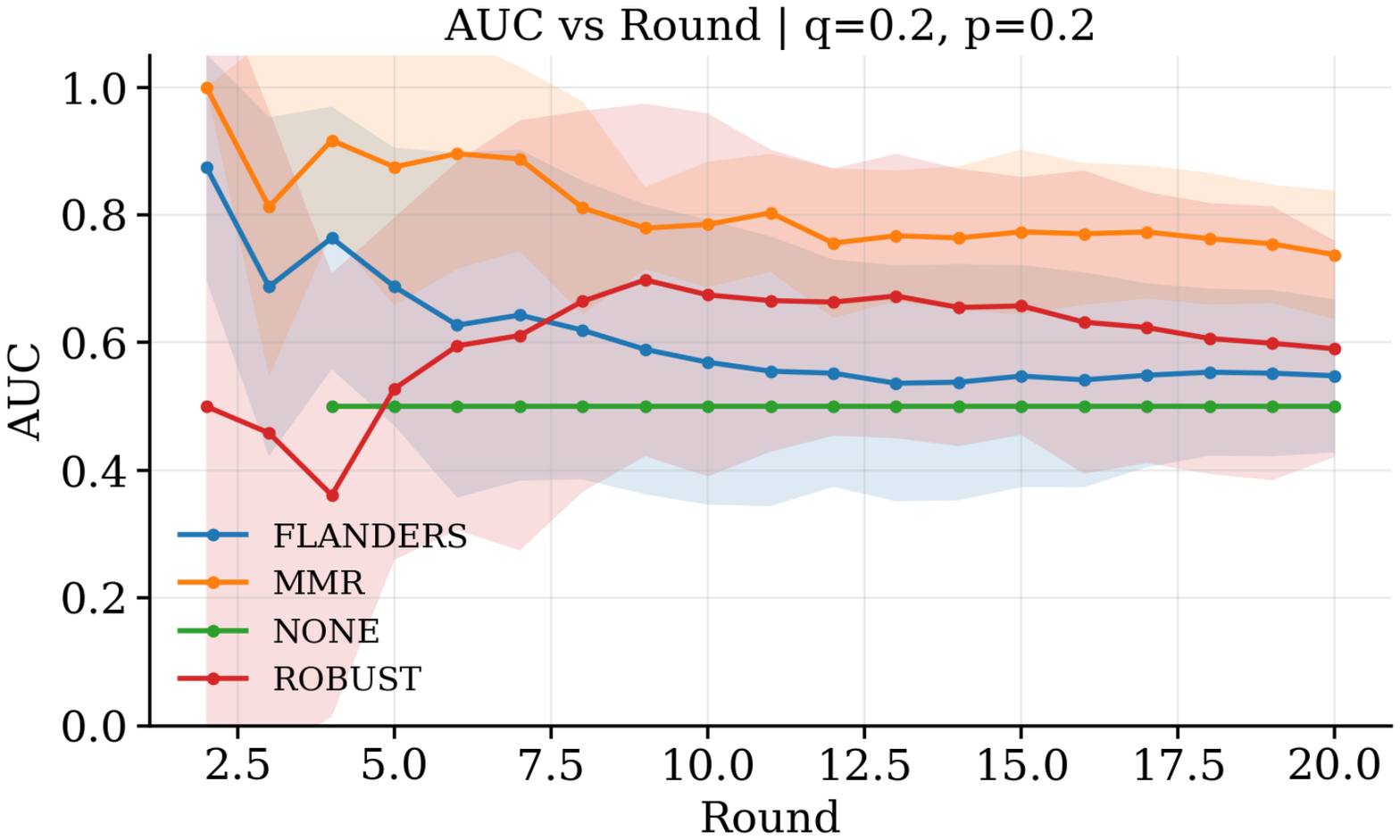}
        \caption{AUC at availability $q=0.2$.}
        \label{fig:auc_q0.2_backdoor}
    \end{subfigure}
    \hfill
    \begin{subfigure}[t]{0.32\textwidth}
        \centering
        \includegraphics[width=\linewidth]{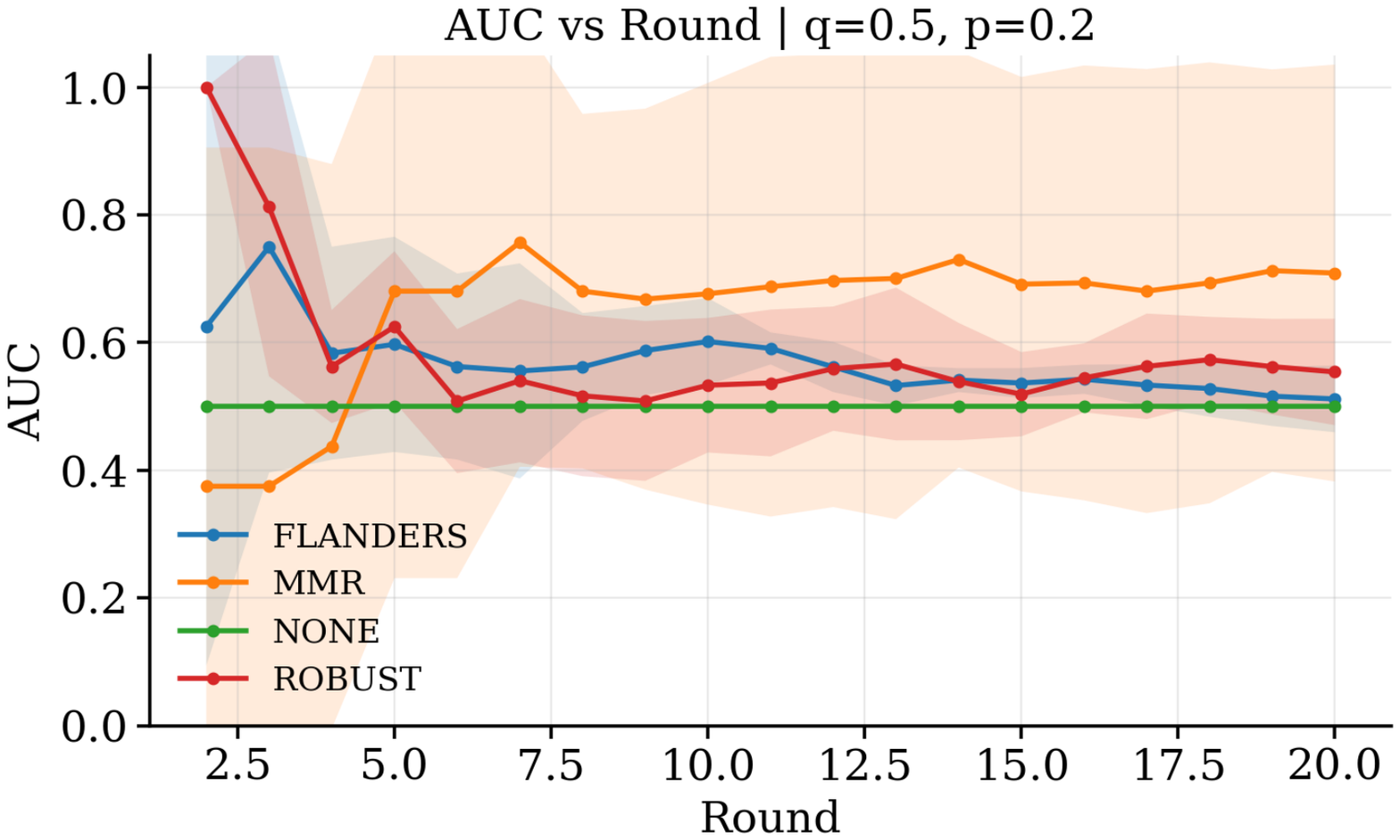}
        \caption{AUC at availability $q=0.5$.}
        \label{fig:auc_q0.5_backdoor}
    \end{subfigure}
    \hfill
    \begin{subfigure}[t]{0.32\textwidth}
        \centering
        \includegraphics[width=\linewidth]{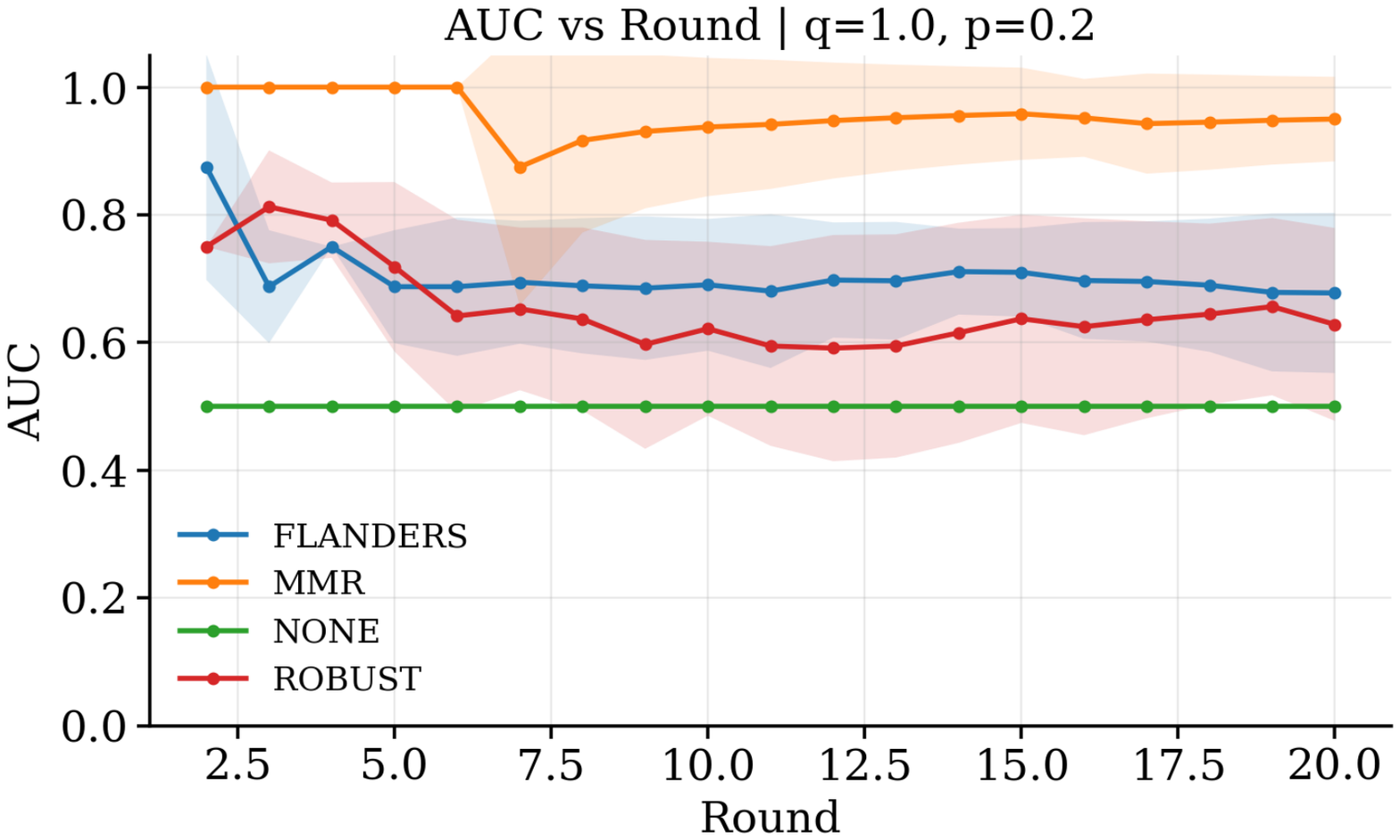}
        \caption{AUC at availability $q=1.0$.}
        \label{fig:auc_q1.0_backdoor}
    \end{subfigure}

    \vspace{0.8em}

    \begin{subfigure}[t]{0.32\textwidth}
        \centering
        \includegraphics[width=\linewidth]{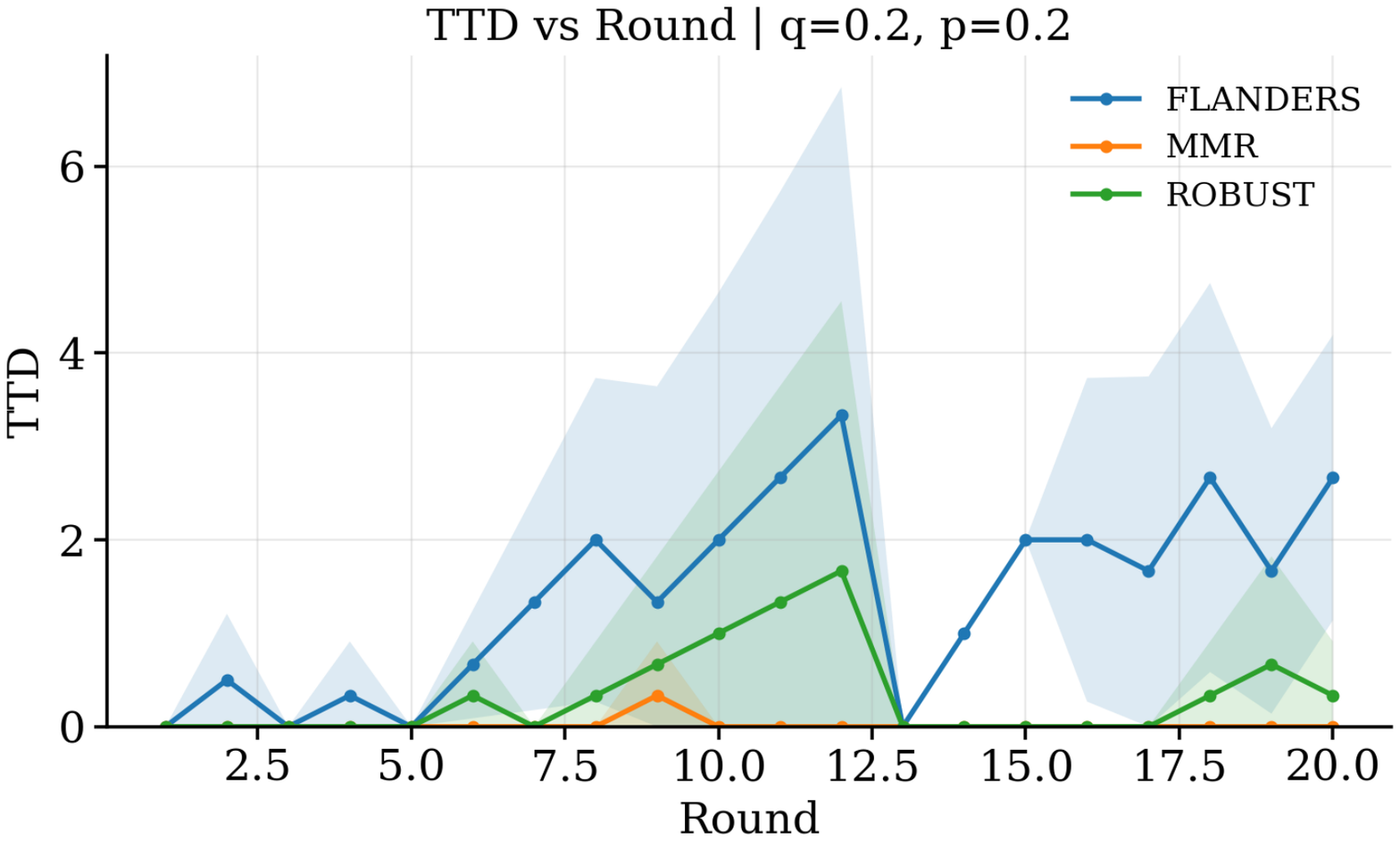}
        \caption{TTD at availability $q=0.2$.}
        \label{fig:ttd_q0.2_backdoor}
    \end{subfigure}
    \hfill
    \begin{subfigure}[t]{0.32\textwidth}
        \centering
        \includegraphics[width=\linewidth]{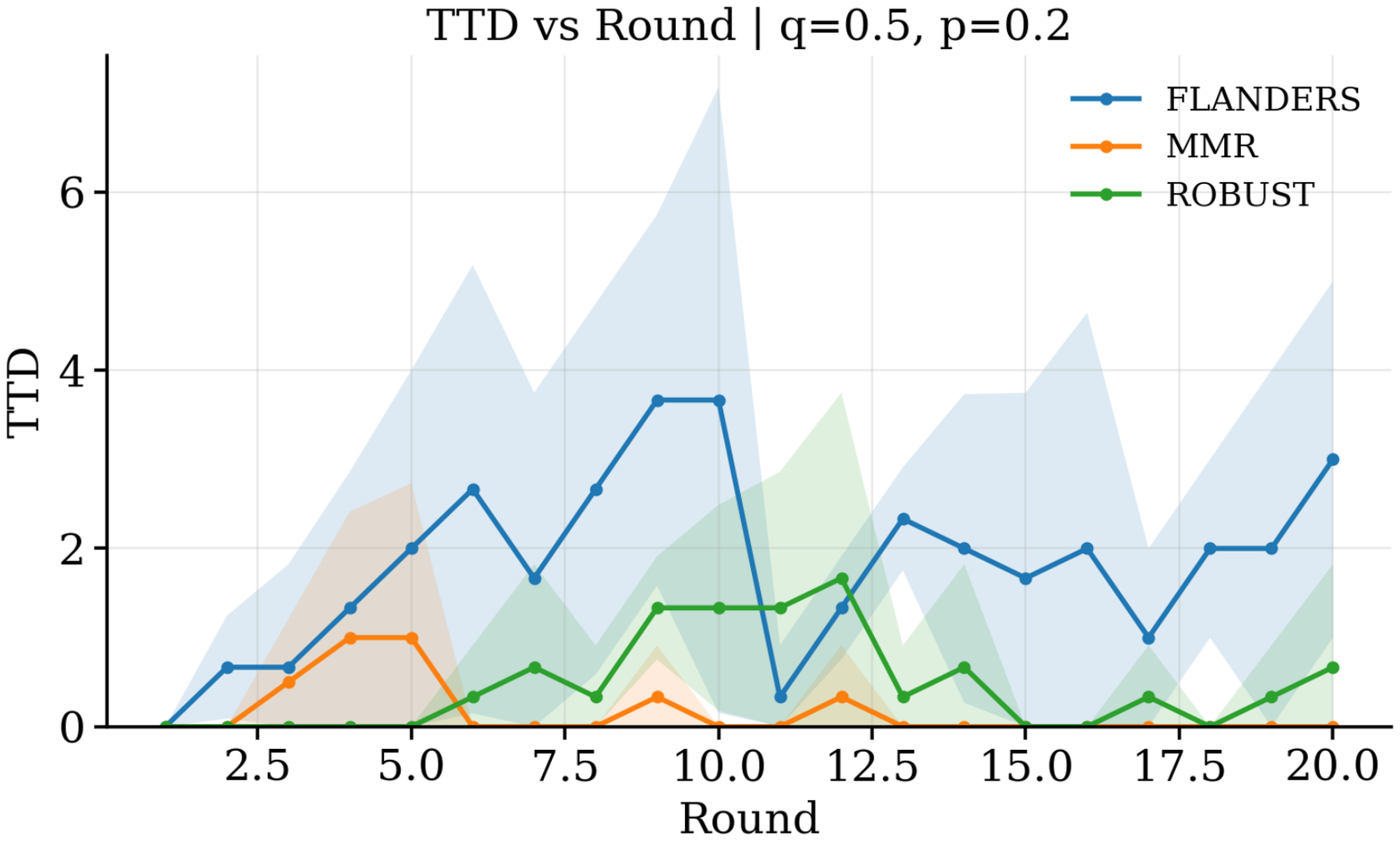}
        \caption{TTD at availability $q=0.5$.}
        \label{fig:ttd_q0.5_backdoor}
    \end{subfigure}
    \hfill
    \begin{subfigure}[t]{0.32\textwidth}
        \centering
        \includegraphics[width=\linewidth]{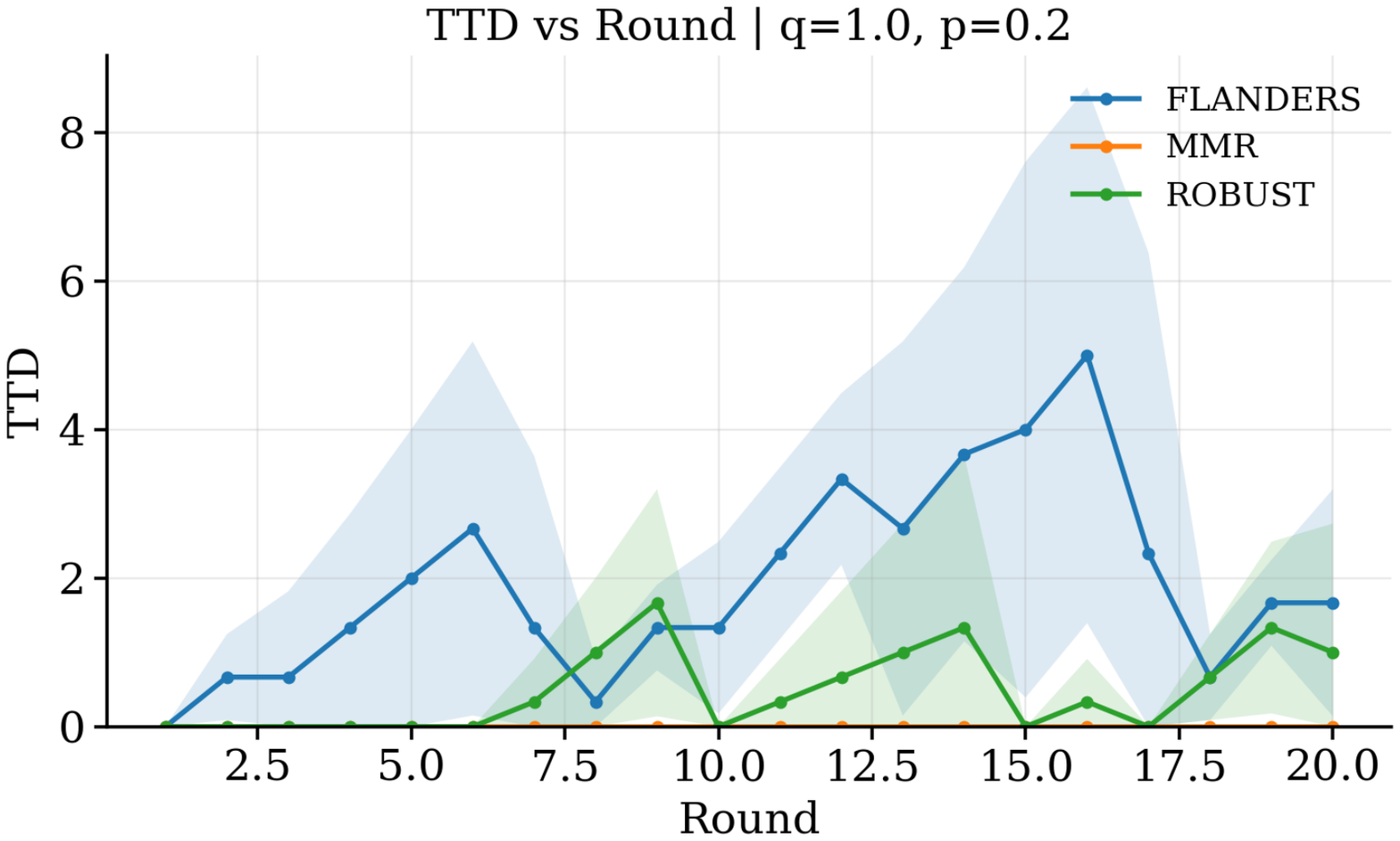}
        \caption{TTD at availability $q=1.0$.}
        \label{fig:ttd_q1.0_backdoor}
    \end{subfigure}

        \vspace{0.8em}

    \begin{subfigure}[t]{0.32\textwidth}
        \centering
        \includegraphics[width=\linewidth]{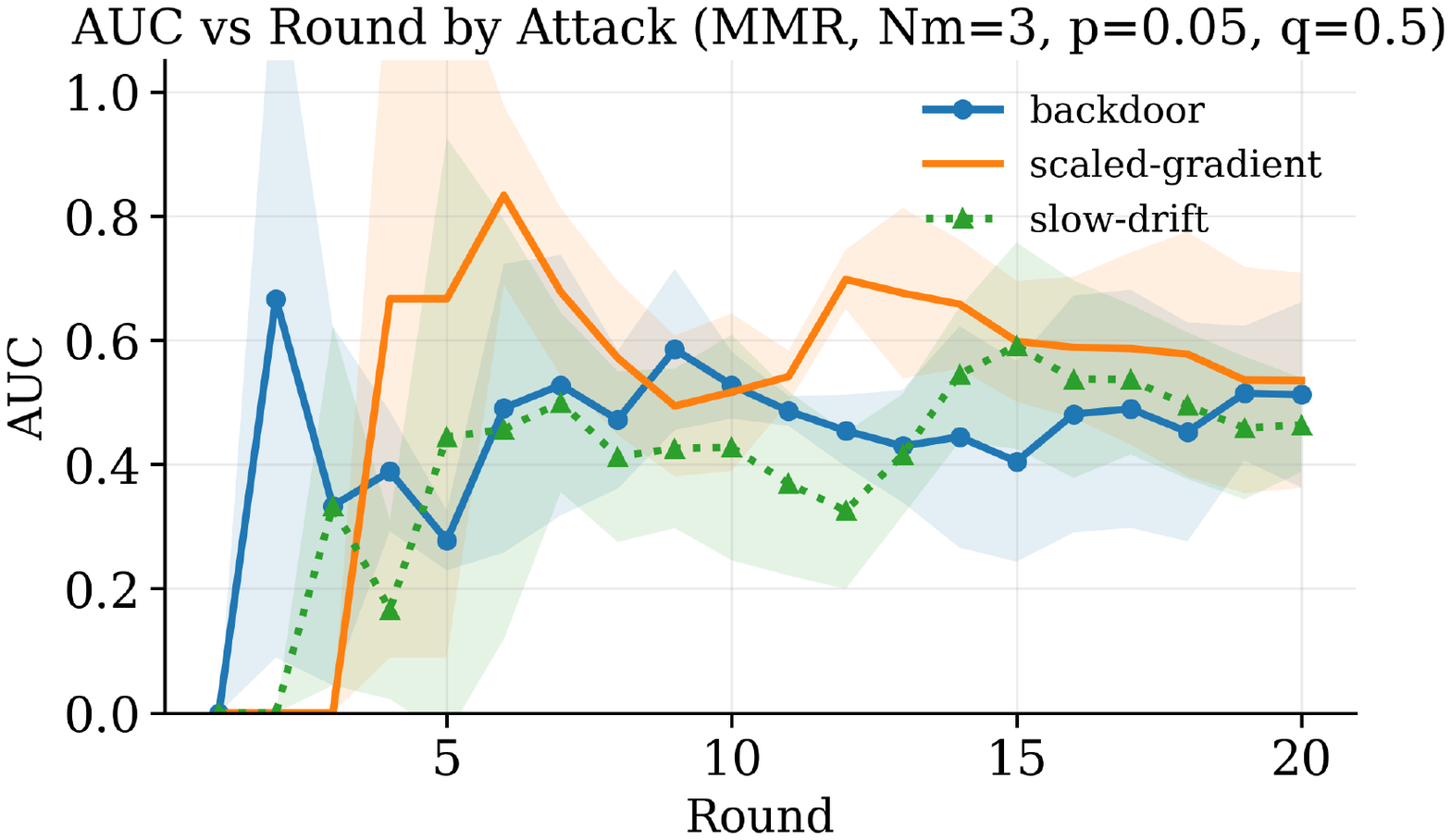}
        \vspace{-1.5em}
        \caption{AUC for attacks at $p=0.05$.}
        \label{fig:mmr_auc_attacks_p0_05}
    \end{subfigure}
        \hfill
    \begin{subfigure}[t]{0.32\textwidth}
        \centering
        \includegraphics[width=\linewidth]{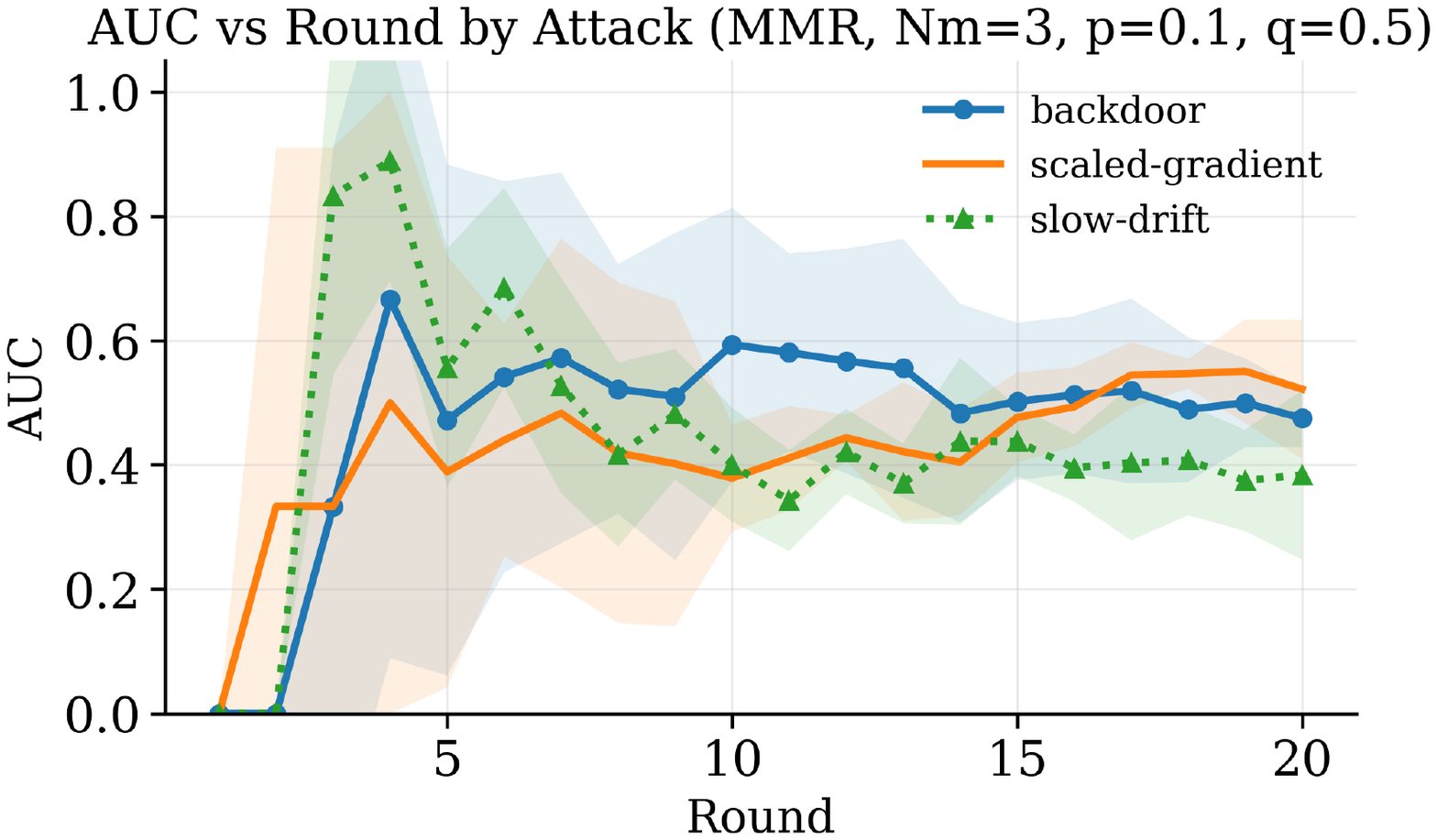}
        \vspace{-1.5em}
        \caption{AUC for attacks at $p=0.1$.}
        \label{fig:mmr_auc_attacks_p0_1}
    \end{subfigure}
        \hfill
    \begin{subfigure}[t]{0.32\textwidth}
        \centering
        \includegraphics[width=\linewidth]{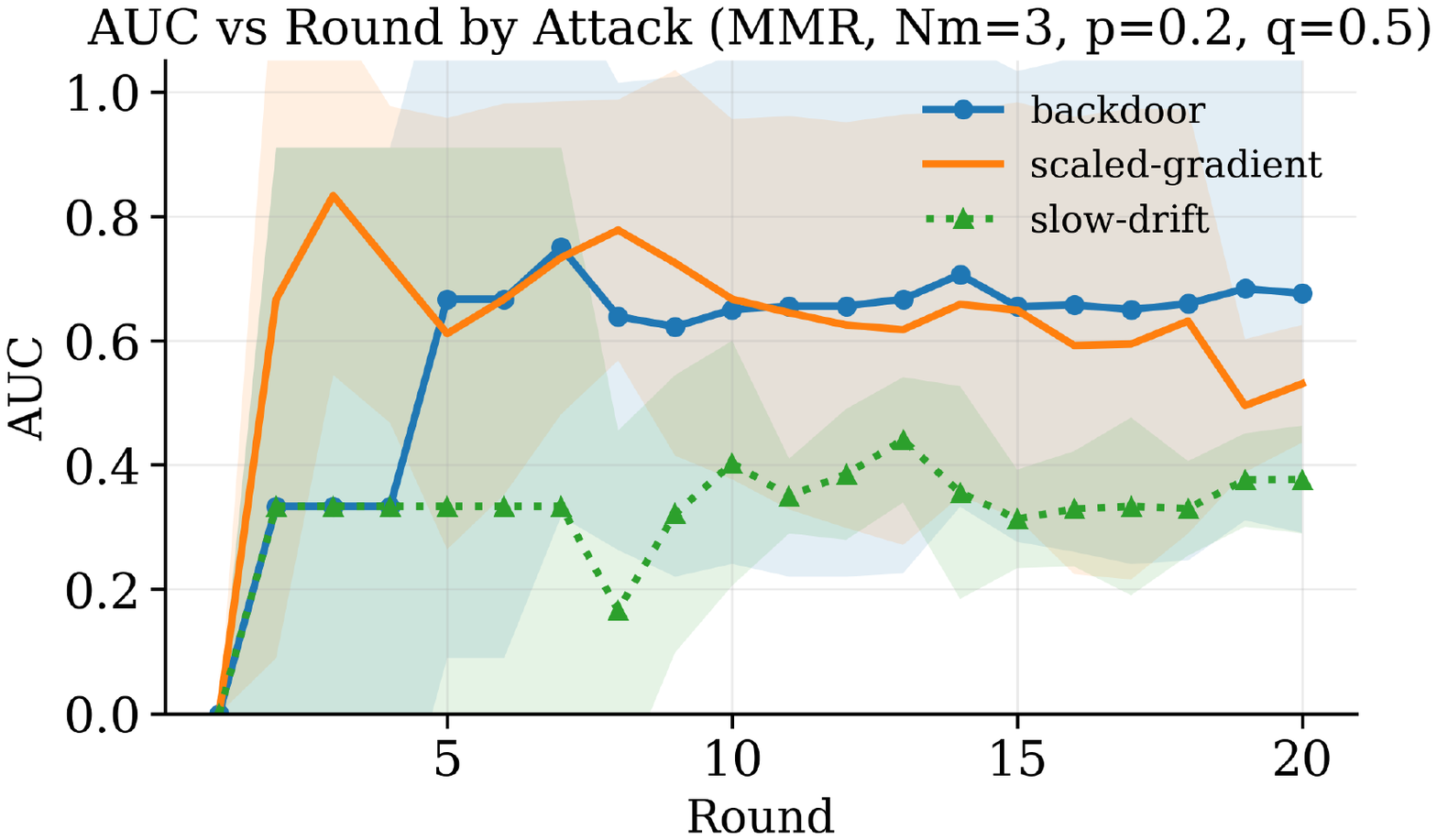}
        \vspace{-1.5em}
        \caption{AUC for attacks at $p=0.2$.}
        \label{fig:mmr_auc_attacks_p0_2}
    \end{subfigure}
    
    \caption{Comparison of MMR, Flanders, Robust, and no-defense baselines under the backdoor attack with setting ($N_m=3$, $q=0.2/0.5/1.0$). The results show complementary views of detection quality: round-wise AUC dynamics, AUC and TTD across client availability levels and different detectors, as well as AUC across different attack scenarios.}
    \label{fig:backdoor_comparison_all}
\end{figure*}

\subsubsection{System overhead and resource consumption.}

Figure~\ref{fig:system_costs_all} compares the runtime and memory costs of MMR, Flanders, Robust, and the no-defense baseline across different availability regimes.
The system overhead results (Figures~\ref{fig:overhead_comparison_q0.2},~\ref{fig:overhead_comparison_q0.5},~\ref{fig:overhead_comparison_q1.0}) show that all methods incur increasing cost over rounds, with clear but moderate differences between them. 
Flanders consistently exhibits the highest total round overhead. 
Robust and the no-defense baseline follow similar trends and remain within the same order of magnitude, indicating that the gap between methods is not extreme. 
In contrast, MMR maintains the lowest overhead across all values of $q$, demonstrating that cross-model comparison introduces only limited additional runtime cost.
The normalized detection cost (Figures~\ref{fig:detection_fraction_q0.2},~\ref{fig:detection_fraction_q0.5},~\ref{fig:detection_fraction_q1.0}) provides further insight into these differences. 
Flanders shows the largest detection-to-total-overhead ratio, indicating that a significant portion of its runtime is spent in the detection stage. 
However, this ratio remains bounded and does not dominate the entire round cost. 
MMR maintains a consistently small ratio across all availability regimes, while Robust remains moderate and the no-defense baseline stays near zero. 
These results indicate that MMR achieves detection with substantially better cost efficiency, while Flanders has higher overhead.
Memory usage (Figures~\ref{fig:memory_comparison_q0.2},~\ref{fig:memory_comparison_q0.5},~\ref{fig:memory_comparison_q1.0}) exhibits a consistent separation across methods. Memory is measured from each process’s /proc/self/status snapshot at every round, using VmRSS as active resident memory and VmSize as virtual address space. 
Flanders shows the highest resident set size (RSS), with memory consumption increasing gradually over rounds, especially at higher availability levels. 
MMR maintains a lower and more stable RSS footprint across all settings. 
Robust and the no-defense baseline lie between these extremes, with relatively flat memory profiles. 
Virtual memory size is comparable across methods, suggesting that the primary differences arise from active memory usage rather than reserved address space.
Overall, MMR remains the most efficient method.

\begin{figure*}[t]
    \centering

    \begin{subfigure}[t]{0.32\textwidth}
        \centering
        \includegraphics[width=\linewidth]{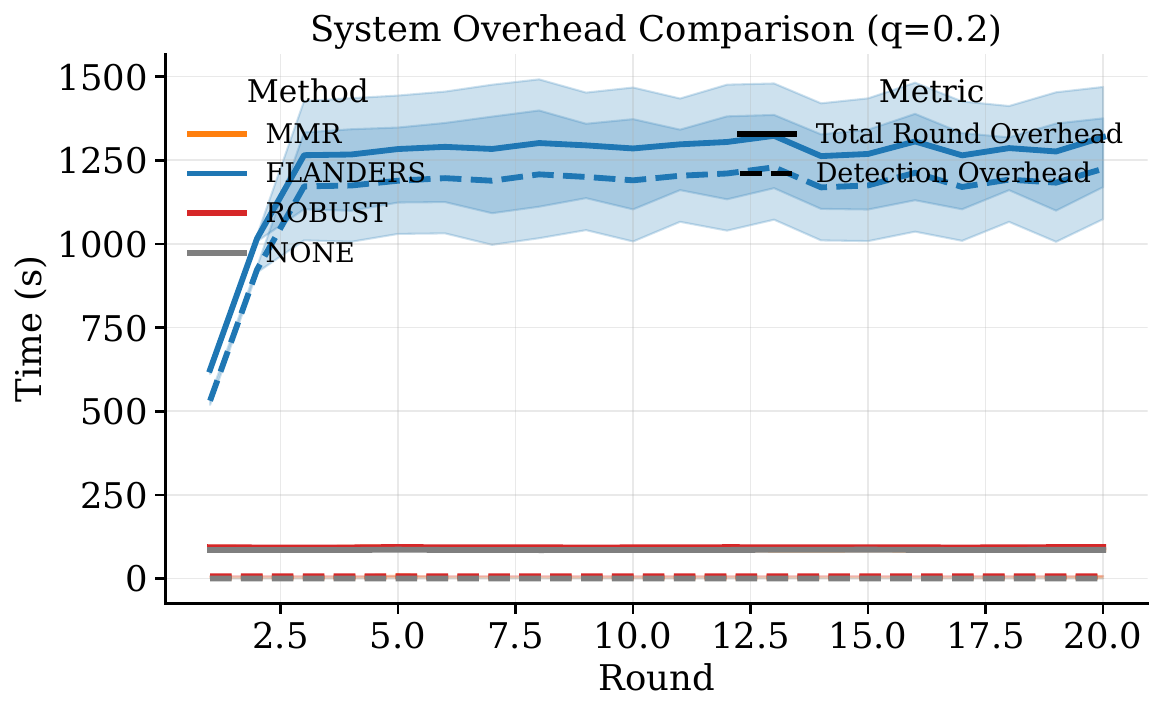}
        \caption{System overhead comparison at q=0.2.}
        \label{fig:overhead_comparison_q0.2}
    \end{subfigure}
    \hfill
    \begin{subfigure}[t]{0.32\textwidth}
        \centering
        \includegraphics[width=\linewidth]{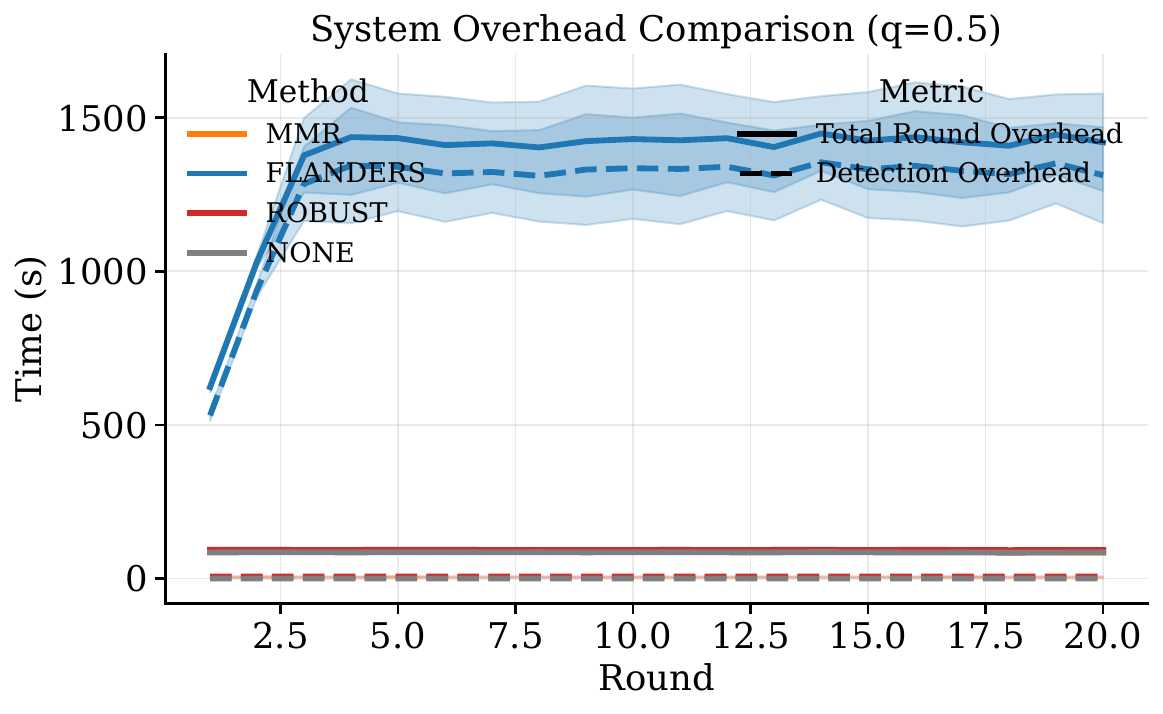}
        \caption{System overhead comparison at q=0.5.}
        \label{fig:overhead_comparison_q0.5}
    \end{subfigure}
    \hfill
    \begin{subfigure}[t]{0.32\textwidth}
        \centering
        \includegraphics[width=\linewidth]{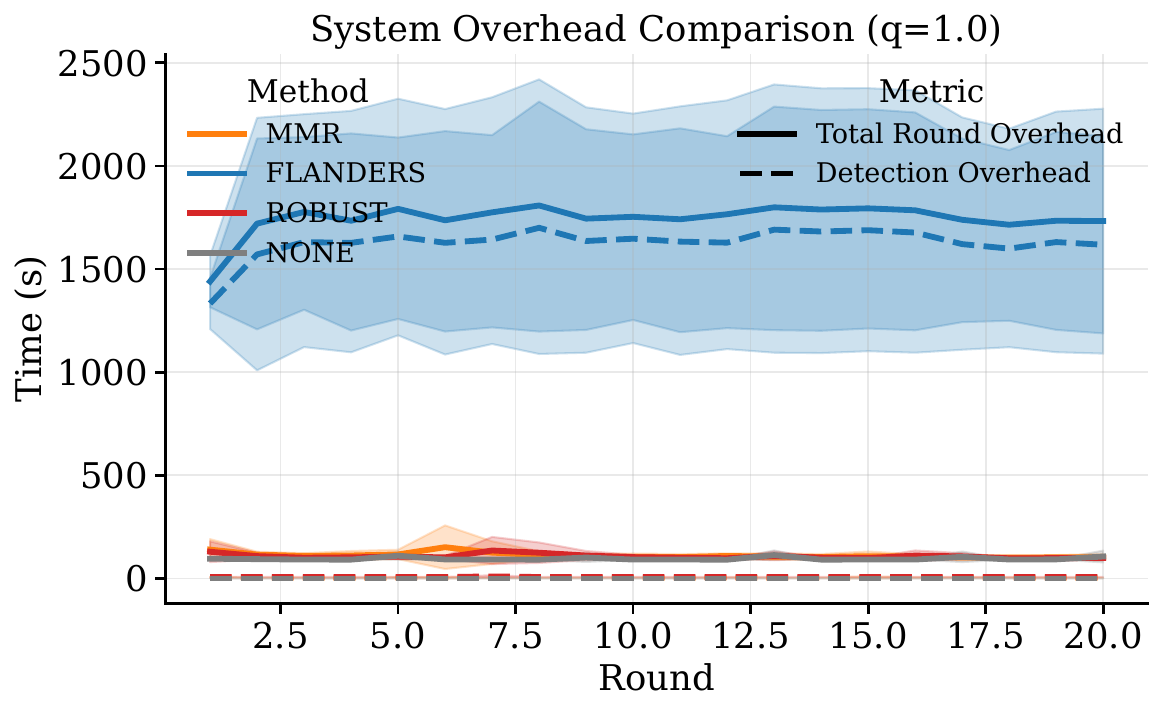}
        \caption{System overhead comparison at q=1.0.}
        \label{fig:overhead_comparison_q1.0}
    \end{subfigure}

       \vspace{0.8em}

       \begin{subfigure}[t]{0.32\textwidth}
        \centering
        \includegraphics[width=\linewidth]{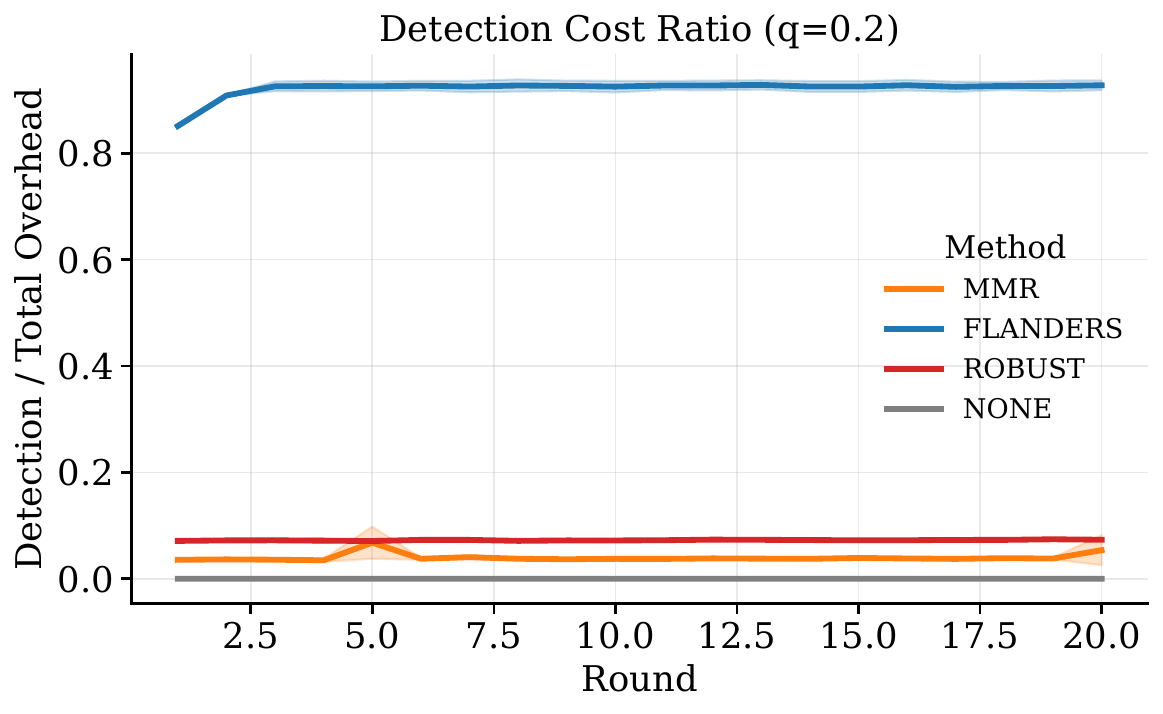}
        \caption{Detection overhead normalized by total round overhead at q=0.2.}
        \label{fig:detection_fraction_q0.2}
    \end{subfigure}
    \hfill
    \begin{subfigure}[t]{0.32\textwidth}
        \centering
        \includegraphics[width=\linewidth]{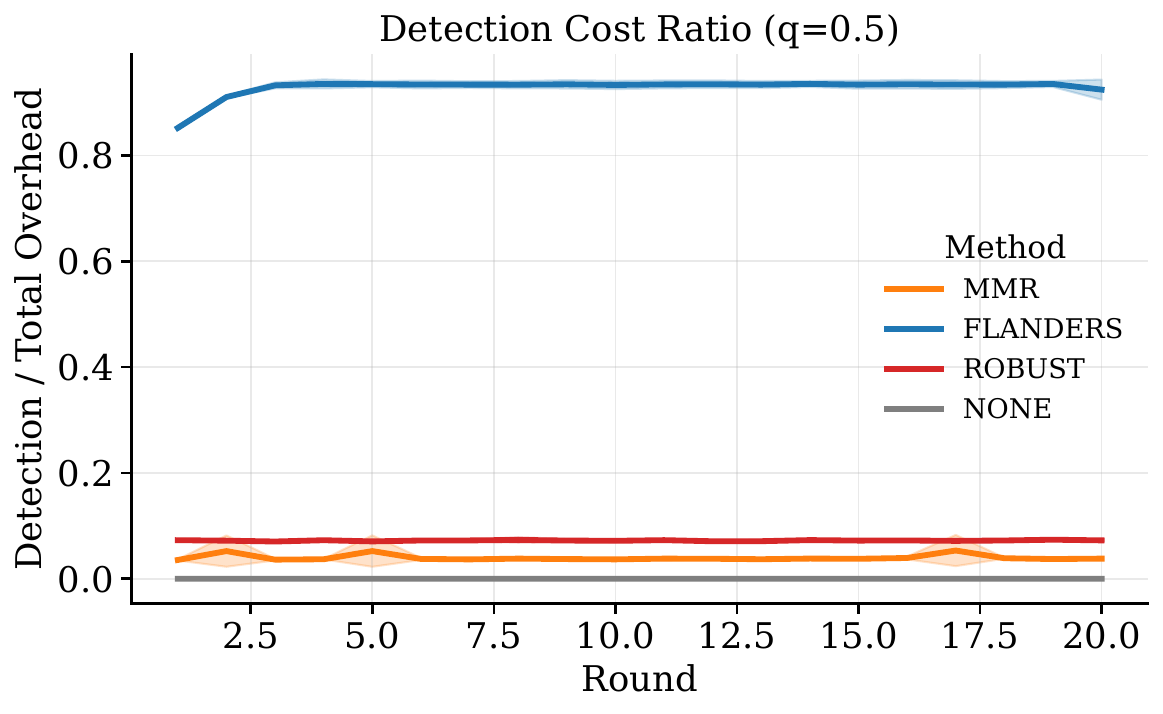}
        \caption{Detection overhead normalized by total round overhead at q=0.5.}
        \label{fig:detection_fraction_q0.5}
    \end{subfigure}
        \hfill
    \begin{subfigure}[t]{0.32\textwidth}
        \centering
        \includegraphics[width=\linewidth]{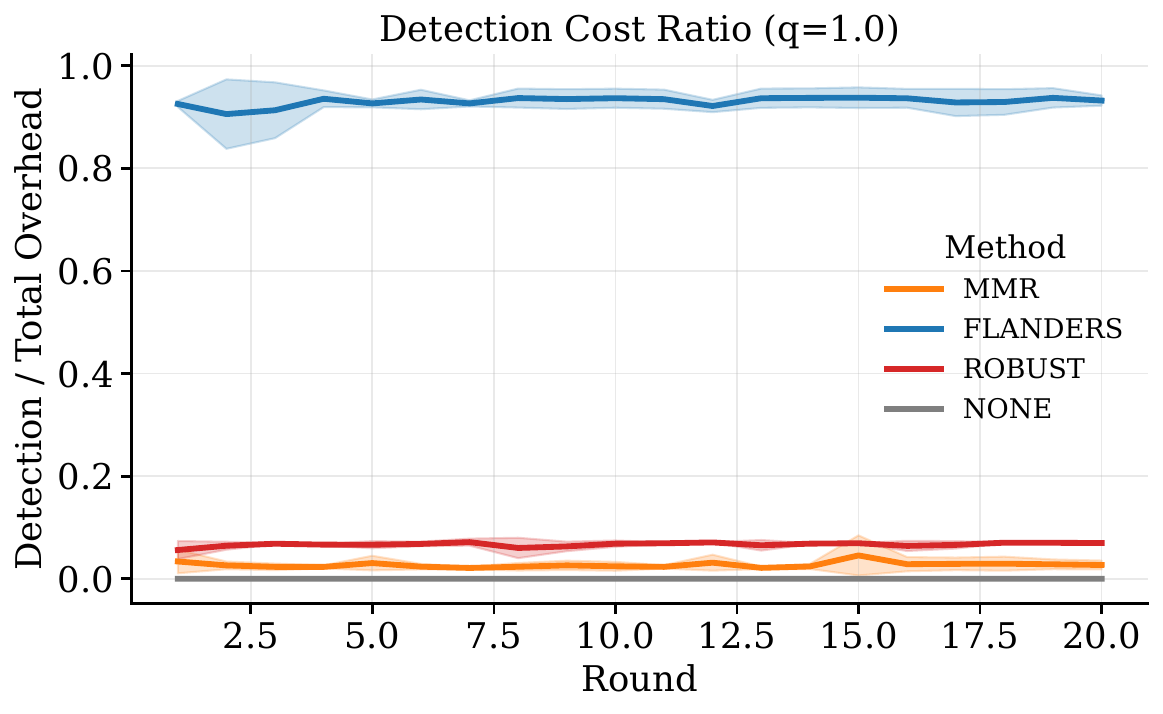}
        \caption{Detection overhead normalized by total round overhead at q=1.0.}
        \label{fig:detection_fraction_q1.0}
    \end{subfigure}

        \vspace{0.8em}

    \begin{subfigure}[t]{0.32\textwidth}
        \centering
        \includegraphics[width=\linewidth]{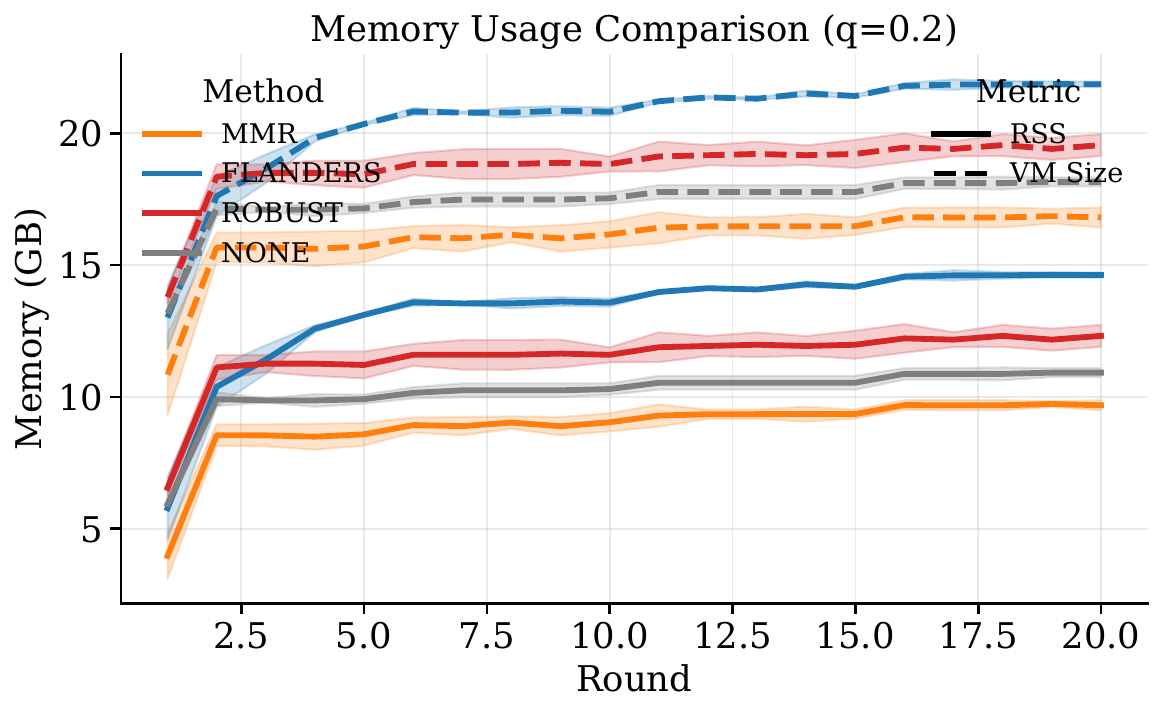}
        \caption{RSS and virtual memory size across methods at q=0.2.}
        \label{fig:memory_comparison_q0.2}
    \end{subfigure}
    \hfill
    \begin{subfigure}[t]{0.32\textwidth}
        \centering
        \includegraphics[width=\linewidth]{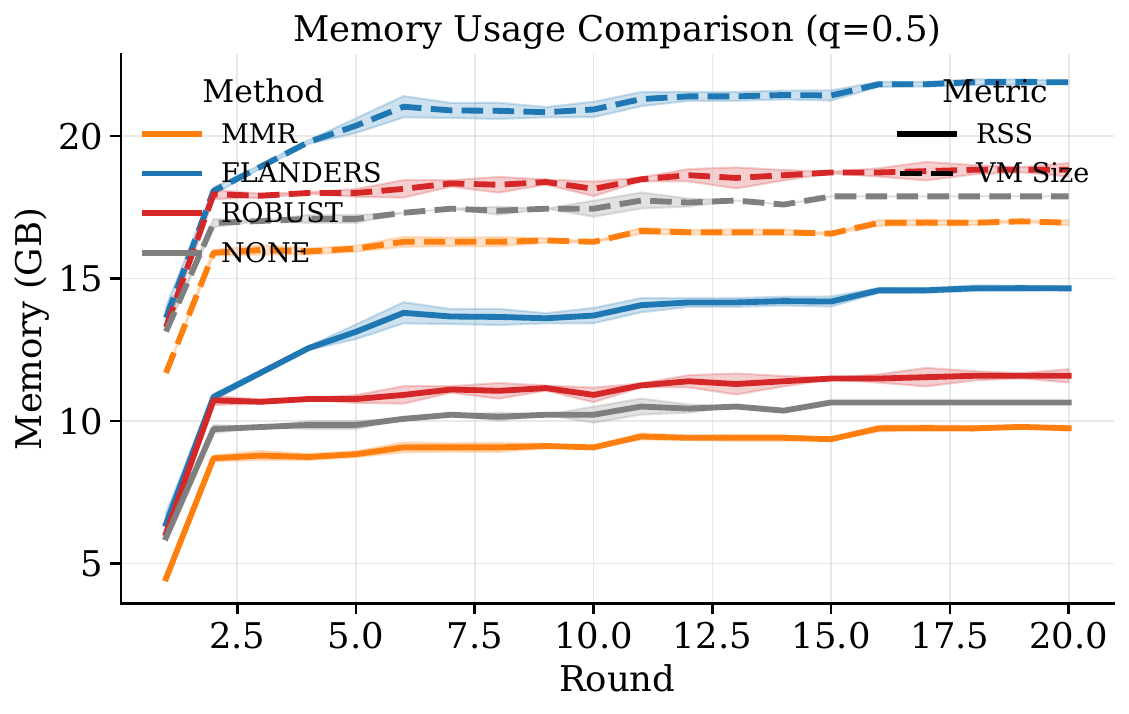}
        \caption{RSS and virtual memory size across methods at q=0.5.}
        \label{fig:memory_comparison_q0.5}
    \end{subfigure}
    \hfill
    \begin{subfigure}[t]{0.32\textwidth}
        \centering
        \includegraphics[width=\linewidth]{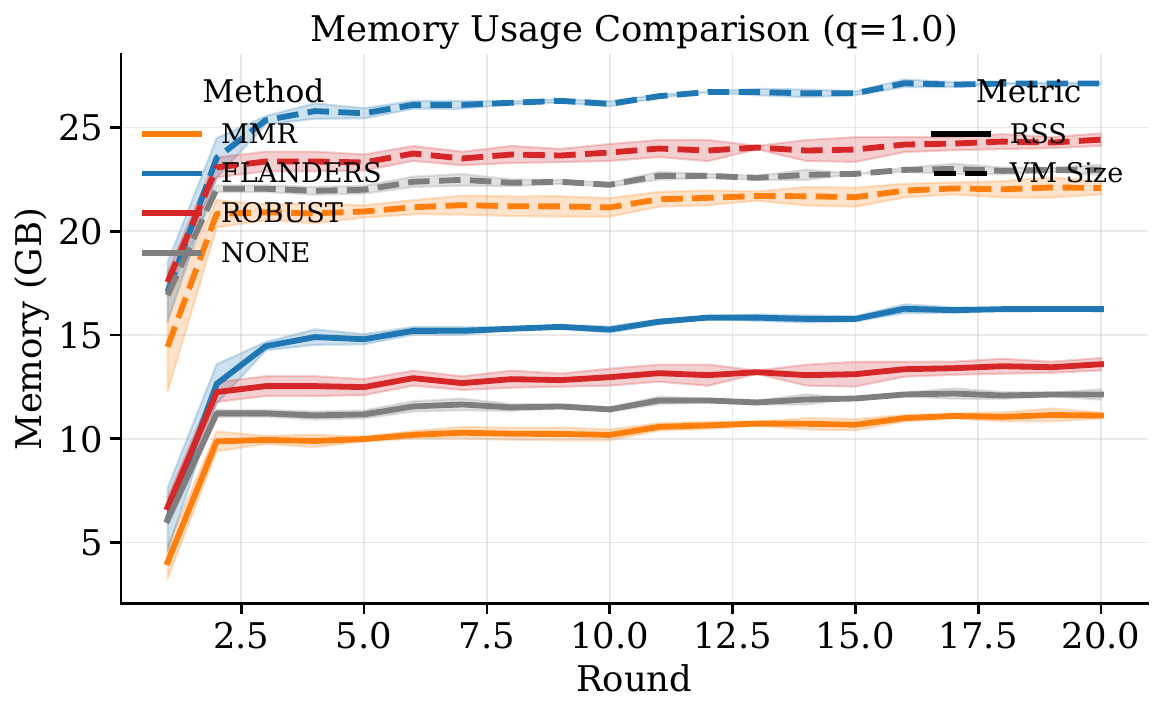}
        \caption{RSS and virtual memory size across methods at q=1.0.}
        \label{fig:memory_comparison_q1.0}
    \end{subfigure}

    \caption{System cost comparison of MMR, Flanders, Robust, and the no-defense baseline. The plots report absolute runtime overhead, normalized detection cost, overall memory footprint. Together, they show that MMR preserves low runtime and memory overhead relative to Flanders and Robust.}
    \label{fig:system_costs_all}
\end{figure*}

\section{Conclusion}
\label{s:conclusion}

This work introduces Multi-Model Rotation as a new defensive method for securing distributed fine-tuning of small language models under intermittent and heterogeneous client participation. 
In contrast to temporal anomaly detectors such as Flanders, which rely on maintaining per-client histories and assume relatively stable participation, MMR leverages multiple lightweight model replicas trained on independently sampled client subsets. 
This design transforms poisoning from a transient perturbation into a persistent cross-model inconsistency: malicious updates inevitably induce divergence in at least one model trajectory.
Our theoretical analysis shows that, under adversarial or intermittent participation, poisoning increases expected inter-model divergence, while limited availability reduces the statistical reliability of temporal aggregation-based detectors. 
Empirically, across availability regimes \(q \in \{0.2, 0.5, 1.0\}\), MMR consistently achieves the strongest detection performance, with higher AUC and faster detection compared to Flanders and Robust. 
In particular, MMR maintains stable detection signals across rounds, whereas Flanders and Robust exhibit weaker and less consistent performance over time, and Robust shows slower adaptation despite moderate stability.
Beyond technical contributions, this work improves security and trustworthiness of federated learning systems, supporting safer deployment of SLMs in sensitive domains like healthcare and IoT.
However, several limitations remain: while our evaluation covers representative attack models (backdoor, slow-drift, and scaled-gradient), it does not exhaustively analyze robustness against fully adaptive or attack-aware adversaries, which remains an important direction for future work. Also, while MMR avoids maintaining per-client histories, it introduces system-level overhead through the maintenance of multiple model replicas; although our experiments show this overhead to be modest, scalability with larger ensembles and highly dynamic participation patterns warrants further study.

\bibliography{references}
\bibliographystyle{apalike}

\section*{Data and Code Availability}

We provide anonymized access to the full implementation of our method:
\begin{center}
\url{https://anonymous.4open.science/r/MMR-BD78}
\end{center}

All experiments use simulated datasets generated within the codebase. The repository includes training and evaluation pipelines, and configuration files required for full reproducibility. All resources are anonymized for double-blind review.

\appendix

\section{Inter-model divergence under poisoning}

\begin{theorem}
\label{thm:divergence_increase}
Consider one synchronous federated aggregation round with a population of $N$ clients.  
Each client $k$ returns a (random) model update $u_k\in\mathbb{R}^d$. Assume a fraction $p\in[0,1]$ of clients are \emph{malicious} and the remaining $1-p$ are honest. Model the update $u$ of a \emph{random} client as a mixture:
\begin{equation}
u = \begin{cases}
u_h = g + \xi_h, & \text{with probability } 1-p \quad(\text{honest}),\\[2pt]
u_m = g + b + \xi_m, & \text{with probability } p \quad(\text{malicious})
\end{cases}
\end{equation}
where
\begin{itemize}
  \item $g\in\mathbb{R}^d$ is the common nominal mean update of honest clients,
  \item $b\in\mathbb{R}^d$ is a fixed \emph{additive bias} introduced by malicious clients (the ``poison''), 
  \item $\xi_h,\xi_m$ are zero-mean random vectors with covariance matrices $\Sigma_h,\Sigma_m$ respectively, independent across clients.
\end{itemize}
Let each model $M_i$ aggregate $n$ independent client updates uniformly sampled from the population (sampling with replacement, or equivalently, $n\ll N$). Define the aggregated per-model update
\begin{equation}
\Delta_i = \frac{1}{n}\sum_{k=1}^n u_{i,k}
\end{equation}
and let two independently sampled models $M_1$ and $M_2$ update their parameters $\theta$ by a step-size $\eta>0$:
\begin{equation}
\theta^1_{t+1}=\theta_t+\eta\Delta_1,\qquad
\theta^2_{t+1}=\theta_t+\eta\Delta_2
\end{equation}
Then the expected squared Euclidean distance (inter-model divergence) between the two models after the round satisfies
\begin{align}
\mathbb{E}\big[ \|\theta^1_{t+1}-\theta^2_{t+1}\|^2 \big]
&= 2\eta^2\cdot\frac{1}{n}\Big( p(1-p)\|b\|^2 + p\operatorname{tr}(\Sigma_m) + (1-p)\operatorname{tr}(\Sigma_h)\Big) \label{eq:thm-main}
\end{align}
Consequently, the net change in expected squared divergence due to poisoning (compared to the benign case $p=0$) is
\begin{equation}
\Delta_{\mathrm{poison}}
\;=\;
\frac{2\eta^2}{n}
\Big(
p(1-p)\|b\|^2
+
p\big(\mathrm{tr}(\Sigma_m)-\mathrm{tr}(\Sigma_h)\big)
\Big).
\end{equation}
Thus, poisoning increases expected divergence whenever
\begin{equation}
(1-p)\|b\|^2 + \mathrm{tr}(\Sigma_m)-\mathrm{tr}(\Sigma_h) > 0.
\end{equation}
\end{theorem}

\textbf{Proof.}
We first compute the first and second moments of a single randomly drawn client update $u$.

\textbf{(1) Mean of a single update.} By the mixture model,
\begin{equation}
\mathbb{E}[u] = (1-p)\mathbb{E}[u_h] + p\mathbb{E}[u_m] = (1-p)g + p(g+b) = g + p b.
\end{equation}

\textbf{(2) Covariance / second moment of a single update.} Compute the (total) variance of $u$:
\begin{equation}
\operatorname{Var}(u) = \mathbb{E}\big[\|u - \mathbb{E}[u]\|^2\big].
\end{equation}
For the honest case, $u_h - \mathbb{E}[u] = g + \xi_h - (g+pb) = -p b + \xi_h$.

For the malicious case, $u_m - \mathbb{E}[u] = g + b + \xi_m - (g+pb) = (1-p)b + \xi_m$.
Thus
\begin{align*}
\operatorname{Var}(u)
&= (1-p)\,\mathbb{E}\big[\| -p b + \xi_h\|^2\big] \;+\; p\,\mathbb{E}\big[\| (1-p)b + \xi_m\|^2\big] \\
&= (1-p)\big( \| -p b\|^2 + \mathbb{E}\|\xi_h\|^2\big) \;+\; p\big( \|(1-p)b\|^2 + \mathbb{E}\|\xi_m\|^2\big)\\
&= (1-p) p^2\|b\|^2 + (1-p)\operatorname{tr}(\Sigma_h) \;+\; p(1-p)^2\|b\|^2 + p\operatorname{tr}(\Sigma_m) \\
&= p(1-p)\|b\|^2 \;+\; (1-p)\operatorname{tr}(\Sigma_h) + p\operatorname{tr}(\Sigma_m).
\end{align*}
We used $\mathbb{E}\|\xi_h\|^2=\operatorname{tr}(\Sigma_h)$ and similarly for $\xi_m$.

\textbf{(3) Variance of the aggregated update $\Delta_i$.} The $n$ updates in the sample used to compute $\Delta_i$ are i.i.d.\ draws from the population model (by sampling with replacement or by assuming $n\ll N$). Hence
\begin{equation}
\operatorname{Var}(\Delta_i) = \operatorname{Var}\Big(\frac{1}{n}\sum_{k=1}^n u_{i,k}\Big) = \frac{1}{n}\operatorname{Var}(u)
\end{equation}

\textbf{(4) Expected squared difference between two independent aggregates.}
Since the two models are sampled independently, $\Delta_1$ and $\Delta_2$ are independent and identically distributed with mean
$\mathbb{E}[\Delta]=\mathbb{E}[u]=g+pb$ and variance
$\operatorname{Var}(\Delta)=\frac{1}{n}\operatorname{Var}(u)$. For independent random vectors $X,Y$ with the same distribution,
\begin{align}
\mathbb{E}\|X-Y\|^2
&= \mathbb{E}\|X\|^2+\mathbb{E}\|Y\|^2-2\mathbb{E}\langle X,Y\rangle \\
&= 2\mathbb{E}\|X\|^2 - 2\langle \mathbb{E}[X],\mathbb{E}[Y]\rangle .
\end{align}
Taking $X=\Delta_1$ and $Y=\Delta_2$, and using independence,
\begin{equation}
\mathbb{E}\langle \Delta_1,\Delta_2\rangle
=
\langle \mathbb{E}[\Delta_1],\mathbb{E}[\Delta_2]\rangle
=
\|\mathbb{E}[\Delta]\|^2 .
\end{equation}
Moreover,
\begin{equation}
\mathbb{E}\|\Delta_1\|^2
=
\operatorname{Var}(\Delta)+\|\mathbb{E}[\Delta]\|^2 .
\end{equation}
Therefore,
\begin{equation}
\mathbb{E}\|\Delta_1-\Delta_2\|^2
=
2\operatorname{Var}(\Delta).
\end{equation}
Substituting $\operatorname{Var}(\Delta)=\frac{1}{n}\operatorname{Var}(u)$ and the expression for $\operatorname{Var}(u)$ derived above gives
\begin{equation}
\mathbb{E}\|\Delta_1-\Delta_2\|^2
=
\frac{2}{n}
\Big(
p(1-p)\|b\|^2
+
p\operatorname{tr}(\Sigma_m)
+
(1-p)\operatorname{tr}(\Sigma_h)
\Big).
\end{equation}

\textbf{(5) Translate to parameter divergence.}
The models update by $\eta\Delta_i$; hence
\begin{equation}
\mathbb{E}\big[\|\theta^1_{t+1}-\theta^2_{t+1}\|^2\big]
=
\eta^2\mathbb{E}\|\Delta_1-\Delta_2\|^2 ,
\end{equation}
which yields the expression \eqref{eq:thm-main}.

To compare the poisoned and benign cases, note that when $p=0$,
\begin{equation}
\mathbb{E}_{\mathrm{benign}}
\big[\|\theta^1_{t+1}-\theta^2_{t+1}\|^2\big]
=
\frac{2\eta^2}{n}\operatorname{tr}(\Sigma_h).
\end{equation}
Therefore, the net change in expected squared divergence under poisoning is
\begin{align}
\Delta_{\mathrm{poison}}
&=
\mathbb{E}_{\mathrm{poison}}
\big[\|\theta^1_{t+1}-\theta^2_{t+1}\|^2\big]
-
\mathbb{E}_{\mathrm{benign}}
\big[\|\theta^1_{t+1}-\theta^2_{t+1}\|^2\big] \\
&=
\frac{2\eta^2}{n}
\Big(
p(1-p)\|b\|^2
+
p\operatorname{tr}(\Sigma_m)
+
(1-p)\operatorname{tr}(\Sigma_h)
-
\operatorname{tr}(\Sigma_h)
\Big) \\
&=
\frac{2\eta^2}{n}
\Big(
p(1-p)\|b\|^2
+
p\big(\operatorname{tr}(\Sigma_m)-\operatorname{tr}(\Sigma_h)\big)
\Big).
\end{align}
Thus, poisoning increases expected inter-model divergence whenever
\begin{equation}
(1-p)\|b\|^2
+
\operatorname{tr}(\Sigma_m)-\operatorname{tr}(\Sigma_h)
> 0.
\end{equation}
In particular, if $\operatorname{tr}(\Sigma_m)\geq \operatorname{tr}(\Sigma_h)$, then
$\Delta_{\mathrm{poison}}>0$ for any $p\in(0,1)$ and $b\neq 0$.

\begin{enumerate}
  \item The theorem isolates the \emph{additive bias} model of poisoning. More general poisoning
  mechanisms, such as multiplicative scaling or targeted gradient manipulation with non-constant
  bias, can be analyzed similarly by replacing $b$ with the appropriate effective bias term and
  tracking the corresponding second-order moments.

  \item The net change in expected divergence has two components: a positive bias-induced term
  $p(1-p)\|b\|^2$ and a covariance-difference term
  $p(\operatorname{tr}(\Sigma_m)-\operatorname{tr}(\Sigma_h))$. Therefore, poisoning does not
  unconditionally increase divergence. Highly coordinated attackers may submit low-variance
  malicious updates with $\operatorname{tr}(\Sigma_m)<\operatorname{tr}(\Sigma_h)$, which can
  partially offset the bias-induced increase.

  \item Under comparable or larger malicious update variance, i.e.,
  $\operatorname{tr}(\Sigma_m)\geq \operatorname{tr}(\Sigma_h)$, the expected divergence increase
  scales as $\frac{1}{n}$ and grows with $p(1-p)\|b\|^2$, peaking in the bias term at
  $p=\tfrac{1}{2}$. Thus, all else equal, smaller per-model sample sizes and moderate malicious
  fractions make divergence-based detection easier in expectation.
\end{enumerate}

Under the same assumptions, if malicious bias $b$ persists across $T$ consecutive rounds and client sampling is independent across rounds, the expected accumulated squared divergence (sum over rounds) grows linearly in $T$, i.e., by $T\cdot\Delta_{\mathrm{poison}}$ in expectation (neglecting second-order effects from checkpoints/rollbacks).

\section{Failure of temporal detection under intermittent participation}

\begin{theorem}
\label{thm:intermittent_failure}
Consider a federated training population of $N$ clients. A fraction $p\in(0,1)$ of these clients are malicious and each malicious client injects an additive bias $b\in\mathbb{R}^d$ (as in Theorem~\ref{thm:divergence_increase}). Time is divided into detection windows of length $T$ rounds.  
Assume \emph{intermittent participation}: every client (honest or malicious) participates in each round independently with probability $q\in(0,1)$ (i.e., availability is i.i.d.\ across rounds and clients). A temporal / per-client detector attempts to flag a client by performing a (scalar) one-sample test on that client's observed updates within the window (projecting updates onto the direction $b$), requiring at least $m\ge 1$ observations of that client in the window to form a decision.
Then:

\begin{enumerate}
  \item For any fixed $m$, the probability that a given malicious client is observed fewer than $m$ times in the window is
  \begin{equation}
    \Pr\big[\#\text{obs} < m\big] \;=\; \sum_{k=0}^{m-1} \binom{T}{k} q^k(1-q)^{T-k}.
  \end{equation}
    In particular, by Markov's inequality,
    \begin{equation}
    \Pr[\#\text{obs} \ge m] \le \frac{qT}{m},
    \quad\text{so}\quad
    \Pr[\#\text{obs} < m] \ge 1 - \frac{qT}{m},
    \end{equation}
    which tends to $1$ when $qT\ll m$.

  \item Suppose the detector projects each observed update $u$ onto the attack direction $b$ (assume $\|b\|>0$) and obtains scalar observations $X_1,\dots,X_K$ for a given client (where $K$ is the realized number of observations in the window). Assume the scalar noise has variance $\sigma^2$ and the malicious mean shift in the projection is $\mu = \langle b, b\rangle/\|b\| = \|b\|$ (so malicious observation mean differs from honest mean by $\mu$ in this projection). Under a standard two-sided $z$-test at level $\alpha$, the \emph{required} number of observations to achieve power $1-\beta$ satisfies
\begin{equation}
    K \;\ge\; \frac{(z_{1-\alpha/2}+z_{1-\beta})^2 \sigma^2}{\mu^2}
\end{equation}
  where $z_{\gamma}$ is the $\gamma$-quantile of a standard normal. Consequently, if the expected number of observations $ \mathbb{E}[K]=qT$ is substantially smaller than this RHS, the per-client detector cannot reliably detect the malicious bias with the desired power and significance.
\end{enumerate}
Therefore, for intermittent participation with small $q$, either (i) many malicious clients are not observed sufficiently often (Item 1), or (ii) the detector needs an impractically large window $T$ to accumulate enough observations for statistical power (Item 2). In both cases, single-model / temporal detectors that rely on repeated per-client observations are likely to fail, motivating ensemble or multi-model designs that expose poisoning via cross-model divergence.
\end{theorem}

\begin{proof}
We prove Items 1 and 2 in turn.

\textbf{Item 1 (Observation scarcity).}  
By assumption each client participates in each round independently with probability $q$. The number of times a given client is observed in the window of length $T$ is therefore a $\mathrm{Binomial}(T,q)$ random variable. The exact expression for the probability of seeing fewer than $m$ observations is the CDF of this binomial:
\begin{equation}
\Pr[\#\text{obs} < m] \;=\; \sum_{k=0}^{m-1} \binom{T}{k} q^k(1-q)^{T-k}.
\end{equation}
This proves the stated identity.
To bound this probability in the sparse-participation regime, we consider the complementary event. Since $K\sim\mathrm{Binomial}(T,q)$ and $\mathbb{E}[K]=Tq$, Markov's inequality gives
\begin{equation}
\Pr[K \ge m] \le \frac{\mathbb{E}[K]}{m} = \frac{Tq}{m}.
\end{equation}
Therefore,
\begin{equation}
\Pr[K < m] \ge 1 - \frac{Tq}{m}.
\end{equation}
When $Tq \ll m$, this lower bound approaches $1$, proving that observation scarcity is severe when expected participation is small relative to the required number of observations.

\textbf{Item 2 (Statistical power scaling).}  
Consider the scalarized detector: project the vector-valued updates onto $b$ (or some informed direction aligned with the attack), yielding real-valued observations. For an honest client, the projected observations have mean $\mu_0$ (nominal); for malicious client they have mean $\mu_0+\mu$ where $\mu=\|b\|>0$ by normalization of the projection. Assume the noise in projections is zero-mean with variance $\sigma^2$ (homoskedastic; this is a standard simplifying assumption — heteroskedastic noise only alters constants).

A classical large-sample $z$-test for the mean difference rejects the null at level $\alpha$ when
\begin{equation}
\frac{\bar{X}-\mu_0}{\sigma/\sqrt{K}} > z_{1-\alpha/2}
\end{equation}
where $\bar{X}$ is the sample mean over $K$ observations. Under the alternative (malicious client) $\bar{X}$ is distributed around $\mu_0+\mu$, so the test statistic's noncentrality is $\mu\sqrt{K}/\sigma$. For the test to have power at least $1-\beta$, the noncentrality must exceed $z_{1-\alpha/2}+z_{1-\beta}$ (standard normal approximation for two-sided test). Thus we need
\begin{equation}
\frac{\mu\sqrt{K}}{\sigma} \;\ge\; z_{1-\alpha/2}+z_{1-\beta}
\end{equation}
which rearranges to
\begin{equation}
K \;\ge\; \frac{(z_{1-\alpha/2}+z_{1-\beta})^2 \sigma^2}{\mu^2}
\end{equation}
This proves the stated scaling.

Now observe the operational consequence: the expected number of observations available for a client in window $T$ is $\mathbb{E}[K]=qT$. If $qT$ is much smaller than the RHS above, the detector cannot reliably detect a malicious client within that window without increasing $T$ (to accumulate more $K$) or increasing participation $q$. For small $q$, reaching the required $K$ requires infeasibly large $T$, which in practice may conflict with responsiveness or model staleness constraints.

Combining Items 1 and 2 yields the claimed dichotomy: either many malicious clients are seldom seen (Item 1), or the detector lacks power unless it waits for a very long window (Item 2). Thus, single-model temporal/per-client detectors fail under intermittent participation.
\end{proof}

\begin{enumerate}
  \item The theorem abstracts the key practical difficulty: \emph{availability-induced sparsity} of observations. Many temporal detectors (e.g., history-based anomaly scores, Flanders-like methods) implicitly assume reasonably frequent client participation so that per-client statistics are estimable; when participation is intermittent, their assumptions break.
  \item Attackers can exploit intermittency strategically: if malicious clients participate only rarely or coordinate to concentrate their updates into short bursts, they can further reduce detectability under per-client tests.  
\end{enumerate}

\section{Response to attacks and further ablation study}

To evaluate the behavior of the proposed MMR framework under sparse participation and weak adversarial presence, we designed a controlled experiment emulating intermittent client availability and low attack intensity. 
The setup consists of $N=100$ simulated clients with participation probability $q=0.2$ per round, such that only approximately one-fifth of the population contributes updates during a given round. 
A small attacker fraction $p=0.05$ was introduced, corresponding to five compromised clients submitting perturbed updates. 
Each client performs a local training epoch using a lightweight transformer-based text model on assigned prompt-based data, ensuring reproducible compute and communication conditions.

Three attack modalities were evaluated to test the detector under different perturbation dynamics:
(1) a \emph{Slow-Drift} attack introducing a persistent low-magnitude bias across rounds;
(2) a \emph{Scaled-Gradient} attack that amplifies local updates by a fixed scaling factor;
and
(3) a \emph{Backdoor} attack that alters a subset of prompt-response patterns to induce targeted generation behavior.
Each attack type was evaluated independently across three random seeds to assess variability.

The ensemble size was varied as $N_m\in\{1,2,3,5\}$ to examine the relationship between model multiplicity and anomaly observability. 
The $N_m=1$ configuration serves as a single-model baseline equivalent to standard federated training without redundancy. 
For multi-model settings, ensembles shared an overlap ratio $\rho=0.2$, allowing partial client reuse across models. 
Throughout training, the server continuously logged probe variance and pairwise divergence statistics together with detector activation and quarantine events. 
After each run, the framework aggregated divergence metadata and anomaly statistics for post-analysis.

\begin{table*}[t]
\centering
\footnotesize
\caption{Detection outcomes for MMR under low participation ($q=0.2$) and sparse attack ratio ($p=0.05$). 
Each configuration was averaged across three random seeds. Reported metrics include mean observed pairwise divergence ($D_{i,j}$) and probe variance statistics extracted from the server logs.}
\begin{tabular}{|c|c|p{2cm}|p{6cm}|}
\hline
\textbf{Attack Type} & \textbf{$N_m$} & \textbf{Mean Pairwise Divergence $D_{i,j}$} & \textbf{Probe Variance $\sigma^2_{\text{probe}}$ / Remarks} \\ \hline

Slow-Drift & 1 & – & No variance spikes; detector inactive \\ \hline
Scaled-Gradient & 1 & – & No variance spikes; detector inactive \\ \hline
Backdoor & 1 & – & No variance spikes; detector inactive \\ \hline

Slow-Drift & 2 & 0.60 (round 1 spike) & Probe variance transient detected \\ \hline
Scaled-Gradient & 2 & 5.90 -- 20.7 (round 2 spikes) & Probe variance $\approx$ 0.17 -- 0.79; intermittent spike alerts \\ \hline
Backdoor & 2 & 3.35 -- 4.84 & Probe variance $\approx$ 1.00; repeated spikes but no confirmed isolation \\ \hline

Slow-Drift & 3 & 99.47 / 82.52 & Two models quarantined; strong divergence event in one seed \\ \hline
Scaled-Gradient & 3 & 0.038 (weak) & Probe variance $\approx$ 0.012; below adaptive threshold \\ \hline
Backdoor & 3 & 0.35 -- 0.59 & Probe variance $\approx$ 0.13 -- 1.00; minor spikes without isolation \\ \hline

Slow-Drift & 5 & 2.94 -- 4.26 & Probe variance $\approx$ 0.00040 -- 0.00062; divergence events in two seeds \\ \hline
Scaled-Gradient & 5 & 3.82 (seed 1 only) & Probe variance spikes $\approx$ 0.0181; no anomalies in remaining seeds \\ \hline
Backdoor & 5 & 16.53 -- 16.85 (seed 1) & Probe variance $\approx$ 0.074 (seed 2); no divergence in seeds 2--3 \\ \hline

\end{tabular}
\label{tab:mmr_divergence}
\end{table*}

\paragraph{Quantitative analysis of divergence and variance signals.}
Table~\ref{tab:mmr_divergence} summarizes the anomaly indicators extracted from server-side logs under sparse participation ($q=0.2$) and low adversarial intensity ($p=0.05$). 
For the single-model baseline ($N_m=1$), all attack types produced stable and statistically unremarkable traces. 
No pairwise divergence statistics exist in this configuration and probe variance remained near baseline, confirming that single-model training lacks comparative structure for disagreement-based anomaly detection.
With $N_m=2$, divergence-related events began to emerge. 
For Scaled-Gradient attacks, divergence values in the range $5.90$--$20.7$ appeared at round~2 together with intermittent probe-variance spike alerts. 
Backdoor attacks produced weaker divergence signals ($3.35$--$4.84$) accompanied by repeated variance spikes, although no stable isolation event was triggered. 
These results indicate that even minimal model redundancy can expose inconsistencies between independently evolving model trajectories that remain invisible in conventional single-model training.

\paragraph{Divergence amplification under larger ensembles.}
At $N_m=3$, anomaly behavior became more heterogeneous across attack types. 
Slow-Drift attacks generated substantially stronger divergence signals ($99.47/82.52$), resulting in two model quarantine events in one seed. 
In contrast, Scaled-Gradient and Backdoor attacks remained comparatively weak, producing only small divergence magnitudes and variance fluctuations below adaptive thresholds. 
This suggests that increasing model multiplicity can amplify persistent perturbation patterns while remaining less sensitive to transient or weak attacks under sparse participation.
For $N_m=5$, divergence observability remained attack dependent. 
Slow-Drift attacks continued to produce moderate divergence events across multiple seeds, while Backdoor attacks generated larger divergence spikes only in isolated runs. 
Scaled-Gradient perturbations produced occasional probe-variance spikes but did not consistently trigger isolation decisions. 
These results indicate that larger ensembles improve anomaly visibility for certain structured perturbations, although intermittent participation still limits reliable detection for weaker attacks.

Overall, the results support the theoretical expectation that cross-model anomaly observability emerges only when multiple concurrent models evolve on partially overlapping client subsets. 
When $N_m=1$, all updates appear statistically benign because no comparative structure exists for disagreement detection. 
As multiplicity increases, pairwise divergence and probe variance begin exposing inconsistencies between model trajectories, enabling the framework to surface anomalous behaviors that remain hidden in standard single-model federated training.

\end{document}